\documentclass{article}

\usepackage{ifthen}
\newboolean{darkMode}
\setboolean{darkMode}{false}
\newboolean{comments}
\setboolean{comments}{false}

\usepackage{wrapfig}     
\ifthenelse{\boolean{comments}}
{
\newcommand{\rnote}[1]{\textcolor{red}{\footnotesize{\bf (Rotem:} {#1}{\bf ) }}}
\newcommand{\enote}[1]{\textcolor{cyan}{\footnotesize{\bf (Efrat:} {#1}{\bf ) }}}
\newcommand{\inote}[1]{\textcolor{magenta}{\footnotesize{\bf (Ilya:} {#1}{\bf ) }}}
}
{
\newcommand{\rnote}[1]{}
\newcommand{\enote}[1]{}
\newcommand{\inote}[1]{}
}

\usepackage[margin=1in]{geometry}
\usepackage{amsmath}    
\usepackage{amssymb}    
\usepackage{amsthm}
\usepackage{mathtools}
\usepackage{amsfonts}
\usepackage{stmaryrd}
\usepackage{enumitem}
\usepackage{physics}
\usepackage{bbm}        
\usepackage{xparse}
\usepackage{xcolor}
\usepackage{mathrsfs} 
\usepackage{graphicx}
\usepackage{pdfpages}
\usepackage{svg}
\usepackage{newfloat}
\usepackage[normalem]{ulem} 
\usepackage{tikz}
\usepackage{wrapfig} 

\usetikzlibrary{arrows.meta, positioning, calc}

\usepackage{subcaption}

\usepackage{mdframed}
\newmdenv[
  skipabove=10pt,
  skipbelow=10pt,
  linewidth=0.8pt,
  linecolor=black,
  roundcorner=4pt,
  backgroundcolor=gray!10,
  innertopmargin=6pt,
  innerbottommargin=6pt,
  innerleftmargin=10pt,
  innerrightmargin=10pt 
]{program}

\newmdenv[
  skipabove=10pt,
  skipbelow=10pt,
  linewidth=0.8pt,
  linecolor=black,
  roundcorner=4pt,
  backgroundcolor=gray!10,
  innertopmargin=6pt,
  innerbottommargin=6pt,
  innerleftmargin=10pt,
  innerrightmargin=10pt
]{protocol}

\numberwithin{equation}{section}

\newtheorem{theorem}{Theorem}[section]
\newtheorem{lemma}[theorem]{Lemma}
\newtheorem{claim}[theorem]{Claim}

\theoremstyle{definition}
\newtheorem{definition}[theorem]{Definition}

\theoremstyle{definition}

\theoremstyle{remark}
\newtheorem*{remark}{Remark}

\newcommand{\N}{\mathbb{N}}
\newcommand{\E}{\mathbb{E}}
\newcommand{\A}{\mathcal{A}}

\newcommand{\I}{\mathcal{I}}
\newcommand{\V}{\mathcal{V}}

\newcommand{\indicator}[1]{\mathbbm{1}_{\left(#1\right)}}
\newcommand{\id}{\mathbbm{I}}

\newcommand{\Ia}{\mathcal{X}}
\newcommand{\Ib}{\mathcal{Y}}
\newcommand{\Oa}{\mathcal{A}}
\newcommand{\Ob}{\mathcal{B}}
\newcommand{\ia}{{x}}
\newcommand{\ib}{{y}}
\newcommand{\oa}{{a}}
\newcommand{\ob}{{b}}

\newcommand{\imap}{\xi}
\newcommand{\omap}{\alpha}
\newcommand{\bellmap}{(\imap,\omap)}
\newcommand{\flexibility}{\vartheta}

\newcommand{\state}{\psi}
\newcommand{\noneg}{\delta}

\newcommand{\Trans}{\mathcal{T}} 
\newcommand{\trans}{\tau} 
\newcommand{\hp}{\gamma}
\newcommand{\Hp}{\Gamma}
\newcommand{\hpDensity}{g}

\newcommand{\preimage}{w}
\newcommand{\image}{z}
\newcommand{\Preimage}{\MakeUppercase{\preimage}}
\newcommand{\Image}{\MakeUppercase{\image}}

\newcommand{\Gen}{\mathsf{Gen}}
\newcommand{\Enc}{\mathsf{Enc}}
\newcommand{\Eval}{\mathsf{Eval}}
\newcommand{\Dec}{\mathsf{Dec}}

\newcommand{\bellscenario}{\Ia,\Ib,\Oa,\Ob}

\newcommand{\bellscenarioInteger}{|\Ia|,|\Ib|,|\Oa|,|\Ob|}
\newcommand{\bellscenarioDistSet}{\Oa\times\Ob\times\Ia\times\Ib}

\newcommand{\playerA}{A}
\newcommand{\playerB}{B}

\newcommand{\compLeakage}{{\kappa}}

\renewcommand{\P}{\mathcal{P}}

\renewcommand\Pr{\mathrm{Pr}}

\newcommand{\bellscenarioShort}{\mathfrak{B}}
\newcommand{\bellInequality}{\mathcal{I}}
\newcommand{\negl}{\mathrm{negl}}

\newcommand{\security}{\lambda}

\newcommand{\mdl}[1]{\mathscr{L}^{\text{M}}_{#1}}
\newcommand{\amdl}[1]{\mathscr{L}^{\text{A}}_{#1}}
\newcommand{\local}{\mathscr{L}}
\newcommand{\quantum}{\mathscr{Q}}
\newcommand{\uc}{\mathscr{P}}
\newcommand{\ns}{\mathscr{N}}
\newcommand{\localcomp}[1]{\local_{#1}^{\mathrm{comp}}}
\newcommand{\quantumcomp}[1]{\quantum_{#1}^{\mathrm{comp}}}

\newcommand{\mdlparamsExplicit}
{    
    l_\flexibility \;=\; \frac{1}{|\mathcal{Y}|}\!\left(1-(|\mathcal{X}|-1)\!\left(\frac{1}{|\mathcal{X}|}+\kappa+\vartheta\right)\!\right) \; ,
    \qquad
    h_\flexibility \;=\; \frac{1}{|\mathcal{Y}|}\!\left(\frac{1}{|\mathcal{X}|}+\kappa+\vartheta\right)  
}
\newcommand{\mdlparams}{l_\flexibility,h_\flexibility}

\newcommand{\qpt}{\mathrm{QPT}}
\newcommand{\cnpal}{\mathrm{CNPA}_\ell}

\newcommand{\accept}{\mathrm{acc}}
\newcommand{\reject}{\mathrm{rej}}
\newcommand{\continue}{\mathrm{cont}}
\newcommand{\flag}{\mathrm{flag}}

\newcommand{\QHE}{\mathtt{QHE}}
\newcommand{\QHEGen}{\mathtt{Gen}}
\newcommand{\QHEEnc}{\mathtt{Enc}}
\newcommand{\QHEEval}{\mathtt{Eval}}
\newcommand{\QHEDec}{\mathtt{Dec}}

\newcommand{\provermeasurement}{B}

\newcommand{\meas}[2]{\provermeasurement_{#1}^{(#2)}}
\newcommand{\measby}{\meas{\ib}{\ob}}

\newcommand{\povm}[1]{\{ \meas{#1}{\ob} : \ob \in \Ob \}}

\newcommand{\povms}{\{ \povm{\ib} : \ib \in \Ib \}}

\newcommand{\polynomialB}{\pi}

\newcommand{\compSignaling}{\kappa_{\mathtt{S}}}

\newcommand{\plrc}{{P_\mathrm{LRC,\security}}}  
\newcommand{\pcomp}{{P_\mathrm{Bell}}}  

\renewcommand{\tr}{\mathrm{tr}}

\definecolor{magenta}{RGB}{255, 0, 255}

\ifthenelse{\boolean{darkMode}}
    {
        \pagecolor[rgb]{0,0,0}\color[rgb]{1,1,1}
        \usepackage{hyperref}

    }
    {
        \usepackage[
        colorlinks=true,
        linkcolor=black,       
        citecolor=black,       
        urlcolor=blue,         
        pdfborder={0 0 0}      
        ]{hyperref}
    
    }

\renewcommand{\O}{\mathcal{O}}

\usepackage{cite} 

\begin{document}

{
    \title{Computational Bell Inequalities}
    \author{
        Ilya Merkulov\thanks{Email: \texttt{ilya.merkulov@weizmann.ac.il}} \quad and \quad
        Rotem Arnon\thanks{Email: \texttt{rotem.arn@weizmann.ac.il}} \\
        \normalsize The Center for Quantum Science and Technology, \\
        \normalsize Weizmann Institute of Science, Rehovot, Israel
    }
    \date{\today}
    \maketitle
}

\begin{abstract}

    We introduce a systematic approach for analyzing device-independent single-prover interactive protocols under computational assumptions. This is done by establishing an explicit correspondence with Bell inequalities and nonlocal games and constructing a \emph{computational space of correlations}.
    We show how computational assumptions are converted to \emph{computational Bell inequalities}, in their rigorous mathematical sense---a hyperplane that separates the sets of classical and quantum verifier-prover interactions.
    We reveal precisely how the nonsignaling assumption in standard device-independent setups interchanges with the computational challenge of learning a hidden input (that we define). 

    We further utilize our fundamental results to study explicit protocols using the new perspective.
    We take advantage of modular tools for studying nonlocality, deriving tighter Tsirelson bounds for single-prover protocols and bounding the entropy generated in the interaction, improving on previous results.
    Our work thus establishes a modular approach to analyzing single-prover quantum certification protocols based on computational assumptions through the fundamental lens of Bell inequalities, removing many layers of technical overhead.
    
    The link that we draw between single-prover protocols and Bell inequalities goes far beyond the spread intuitive understanding or known results about ``compiled nonlocal games''; Notably, it captures the exact way in which the correspondence between computational assumptions and locality should be understood also in protocols based on, e.g., trapdoor claw-free functions (in which there is no clear underlying nonlocal game). 

\end{abstract}

\newpage
\tableofcontents

\newpage
\section{Introduction}

    Quantum technology paves the way for stronger forms of computation, communication and cryptography.
    With this development, a central question arises~-- how can we verify that the quantum devices used in the above mentioned tasks are actually doing what we want them to do? 
    A simple setup to have in mind is the following: a client with a classical computer (e.g., our current laptops) connects to a server to execute a computation on a quantum computer.
    The server returns the result of the quantum computation and now the client, who has only classical means, wishes to verify that the result of the computation is correct.
    Another example in the realm of cryptography is a setup in which a client interacts classically with a quantum server in order to produce a secret random string~-- a key for further cryptographic applications.
    After producing the alleged key, the client wants to be sure that the key is safe to use and that no one else knows the key.
    
    Such processes exemplify what is known as ``classical verification of quantum device'' in the literature.
    The aim is to verify a property of an uncharacterized quantum device or its result without needing knowledge, nor trust, in its internal workings.

    \subsection*{Bell inequalities}

    The field of verification of quantum devices using classical means dates back to the 60's with the transformative work of Bell~\cite{bell64}.
    The observation made was that when taking two spatially separated devices and using them to ``play a game'', quantum devices that share entanglement can win the game with a probability strictly larger than that of \emph{any} two classical devices.
    Thus, observing a winning probability above the optimal classical one acts as a certificate for the ``quantumness'' of the behavior exhibited by the devices.\footnote{More precisely, the certificate rules out local behavior.}
    A prominent class of tests that certify non-classical correlations are \emph{nonlocal games}~\cite{brunner2014bell}.
    These are interactive protocols in which devices that succeed with higher-than-classical probability are said to violate a Bell inequality.
    Following the original works, it was further shown that the violation of a Bell inequality can be used not only to show that the devices must be quantum, but also that they are generating a secret key in the strongest form of cryptographic standards, a notion known as device-independent security~\cite{pironio2009collective, acin2016certified, acin2016optimal,ekert2014ultimate}.
    Furthermore, such violations can even characterize the quantum state and measurement used by the devices, a process called self-testing~\cite{supic2020self,metger2021self,arnon2018noise,bamps2015sum}.

    \begin{wrapfigure}{r}{0.38\textwidth}
        \centering
        \vspace{-2.5em}  
        \includegraphics[width=0.36\textwidth]{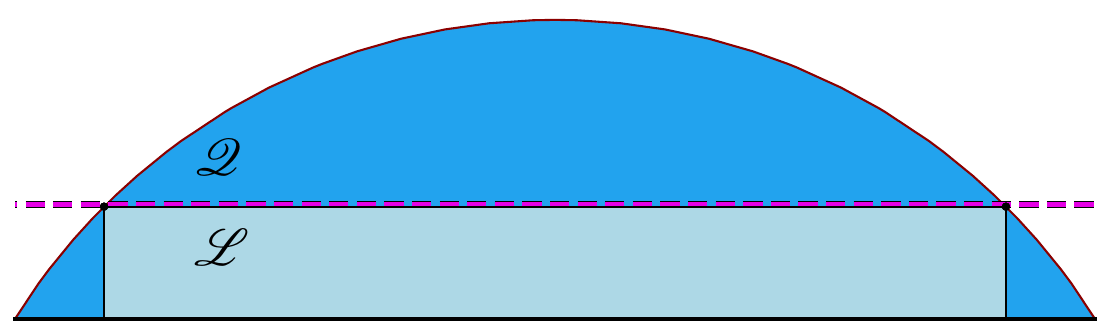}
        \caption{
            \footnotesize
            A 2D slice of the SoC of conditional probability distributions for some nonlocal game.
            The local set~$\local$ in lighter blue and the quantum set~$\quantum$ in darker blue.
            The Bell inequality is represented by the magenta dashed line.
        }
        \label{fig:sets}
        \vspace{-0.8\baselineskip} 
    \end{wrapfigure}

    To better understand what a Bell inequality is, one needs to mathematically define what is meant by the word ``behavior'' used above.
    To this end, consider the set of \emph{correlations}, or conditional probability distributions, that describe quantum devices in a nonlocal game.
    We denote by $P(\oa,\ob \mid \ia,\ib)$ the distribution that describes the probability of the two separated (but potentially entangled) devices, outputting $\oa,\ob$ when given the inputs $\ia,\ib$. 
    By fixing a distribution over the inputs\footnote{The distribution over the inputs is defined by the nonlocal game; in most cases it is simply independent and uniform.} $P(\ia,\ib)$, one can discuss the distribution $P(\oa,\ob,\ia,\ib) = P(\oa,\ob \mid \ia,\ib) P(\ia,\ib)$.
    The set of all distributions $P(\oa,\ob\mid\ia,\ib)$ that can arise using classical devices is called the local set $\local$.
    Similarly, the set of quantum distributions is denoted by $\quantum$.


    A generalization~\cite{scarani2019formalizing} of a Theorem due to Fine~\cite{fine1982hidden}, shows that the set of local (classical) correlations~$\local$ forms a convex polytope.
    The set of quantum correlations~$\quantum$, in contrast, is convex but not a polytope.
    Any classical model can be simulated by a quantum system, and thus~$\local \subset \quantum$.
    A Bell inequality defines a hyperplane that bounds the local set of classical correlations; see Figure~\ref{fig:sets}.
    The ability to violate a Bell inequality, thus certifying quantumness, is based on the existence of a distribution $P^{\star}=P^{\star}(\oa,\ob|\ia,\ib)$ such that $P^{\star} \in \quantum$ and $P^{\star}\notin \local$. The violation was also experimentally verified and recognized with a Nobel Prize in 2022.\footnote{See the \href{https://www.nobelprize.org/prizes/physics/2022/summary/}{Nobel Prize official website}.}

    Investigating the \emph{space of correlations} (SoC) modeled by the probability distributions $P(\oa,\ob,\ia,\ib)$ is both insightful and fruitful.
    Numerous studies considered the form of the SoC in high dimensions, different tools for approximating the quantum set and optimizing different objective functions over the SoC, the relation to other related sets such as the nonsignaling set and the almost-quantum set, foundational aspects of certifiable entanglement, and much more.
    Notably, using the SoC and the tools developed for studying it significantly advanced device-independent cryptography, with the key insight being that it is simpler to consider the SoC broadly instead of certifying a specific apparatus behavior.

    \subsection*{Certification of a single device}

    In recent years a new question arose: is it possible to certify a quantum behavior using only a single device?
    Nonlocal games are games played using two spatially separated, nonsignaling, devices.
    With just a single device, any quantum correlation can be simulated classically.
    Indeed, Bell inequalities, as discussed above, are facets in the space of bipartite (or more) correlations.
    A different approach was then needed.

    A breakthrough came with the results of Mahadev~\cite{mahadev2018classical} and Brakerski et al.~\cite{brakerski2021cryptographic}, who proposed a method to certify the quantumness of a single uncharacterized device using cryptographic techniques.
    This initiated a growing line of research, building on computational assumptions to enable single-device certification across a variety of tasks.\footnote{See for example~\cite{gheorghiu2019computationally,brakerski2020simpler,metger2021self,metger2021device,vidick2021classical,kahanamoku2022classically,liu2022depth,gheorghiu2022quantum,brakerski2023simple,natarajan2023bounding,aaronson2023certified} and references therein.}  
    The main idea of the new approach was to add computational assumptions into the picture.
    Instead of having two computationally unlimited devices, the interaction is now with a single but computationally limited device; The device can now only apply efficient operations and run for a polynomial amount of time.
    While the internal structure of the device remains uncharacterized, it is modeled as a~$\qpt$ system (i.e., a quantum polynomial-time machine).

    While quantum computers are expected to outperform classical ones, many computational problems are still believed to be intractable even for quantum algorithms.
    These conjectured hardness assumptions underpin post-quantum cryptography.
    A prominent example is the Learning with Errors (LWE) problem, which underlies constructions of cryptographic primitives such as trapdoor claw-free functions.
    An efficient classical client, or verifier, can now use a post-quantum cryptographic task to test the quantum device, or server, and see that it behaves as expected.\footnote
    {
        For readers not familiar with these lines of work, in Appendix~\ref{ap_sec:protocol} we present for completeness two concrete examples of such protocols, presented in~\cite{kahanamoku2022classically,kalai2022compiled}.
    }
    The core idea of using a computational assumption was used in the literature to certify quantumness, randomness and key generation, self-testing and verification of computation. 

    Another avenue used to switch from two quantum devices to a single $\qpt$ device is via a concept called compiled nonlocal games~\cite{kalai2022compiled,compiledtrapdoor2024}.
    The main idea is to start with a nonlocal game, for which we know how to certify a quantum behavior, and then mask it using a computationally hard task.
    Then, with the nonlocal game being ``computationally hidden'' we can ask a single device to play the role of the original two devices one after the other.
    In this line of works it is pretty clear that the properties of the nonlocal game are what allow for certification also in the single-device setup.
    Nevertheless, making precise and quantitative statements still requires a lot of work.
   
    \subsection*{Motivation}

    \begin{figure}[t]
        \centering

        \newcommand{\boxheight}{2}
        \newcommand{\boxwidth}{2}
        \newcommand{\nodedistance}{1.5}
        \newcommand{\inputheight}{0.7}   
        \newcommand{\outputheight}{0.7}  
        \newcommand{\arrowgap}{2pt}      
        \newcommand{\arrowspacing}{0.45} 
        \newcommand{\textSize}{\Large}

        \begin{subfigure}[t]{0.47\textwidth}
        \centering
        \begin{tikzpicture}[>=Latex, node distance=\nodedistance cm]

        \node[draw, rounded corners=8pt, minimum width=\boxwidth cm, minimum height=\boxheight cm, align=center] (A) at (0,0) {\textSize$\P$};
        \node[draw, rounded corners=8pt, minimum width=\boxwidth cm, minimum height=\boxheight cm, align=center, right=of A] (B) {\Large$\V$};

        \node[opacity=0] at ([yshift=\inputheight cm + 0.2cm]A.north) {$x$};
        \node[opacity=0] at ([yshift=\inputheight cm + 0.2cm]B.north) {$y$};
        \node[opacity=0] at ([yshift=-\outputheight cm - 0.2cm]A.south) {$a$};
        \node[opacity=0] at ([yshift=-\outputheight cm - 0.2cm]B.south) {$b$};

        \draw[->, opacity=0, shorten >=\arrowgap, shorten <=\arrowgap] ([yshift=\inputheight cm]A.north) -- (A);
        \draw[->, opacity=0, shorten >=\arrowgap, shorten <=\arrowgap] ([yshift=\inputheight cm]B.north) -- (B);
        \draw[->, opacity=0, shorten >=\arrowgap, shorten <=\arrowgap] (A) -- ([yshift=-\outputheight cm]A.south);
        \draw[->, opacity=0, shorten >=\arrowgap, shorten <=\arrowgap] (B) -- ([yshift=-\outputheight cm]B.south);

        \foreach \i in {0,1,2,3} {
        \pgfmathsetmacro{\offset}{0.7 - \i * \arrowspacing}
        \ifodd\i
            \draw[->, shorten >=\arrowgap, shorten <=\arrowgap] 
            ([yshift=\offset cm]B.west) 
            -- ([yshift=\offset cm]A.east);
        \else
            \draw[->, shorten >=\arrowgap, shorten <=\arrowgap] 
            ([yshift=\offset cm]A.east) 
            -- ([yshift=\offset cm]B.west);
        \fi
        }

    \end{tikzpicture}
        \vspace{-0.6em}
        \caption{\footnotesize Two abstract systems with alternating communication}
        \end{subfigure}
        \hfill
        \begin{subfigure}[t]{0.47\textwidth}
        \centering
        \begin{tikzpicture}[>=Latex, node distance=\nodedistance cm]

        \node[draw, rounded corners=8pt, minimum width=\boxwidth cm, minimum height=\boxheight cm, align=center] (A) at (0,0) {\textSize$\playerA$};
        \node[draw, rounded corners=8pt, minimum width=\boxwidth cm, minimum height=\boxheight cm, align=center, right=of A] (B) {\textSize$\playerB$};

        \node at ([yshift=\inputheight cm + 0.2cm]A.north) {$x$};
        \node at ([yshift=\inputheight cm + 0.2cm]B.north) {$y$};

        \node at ([yshift=-\outputheight cm - 0.2cm]A.south) {$a$};
        \node at ([yshift=-\outputheight cm - 0.2cm]B.south) {$b$};

        \draw[->, shorten >=\arrowgap, shorten <=\arrowgap] ([yshift=\inputheight cm]A.north) -- (A);
        \draw[->, shorten >=\arrowgap, shorten <=\arrowgap] ([yshift=\inputheight cm]B.north) -- (B);
        \draw[->, shorten >=\arrowgap, shorten <=\arrowgap] (A) -- ([yshift=-\outputheight cm]A.south);
        \draw[->, shorten >=\arrowgap, shorten <=\arrowgap] (B) -- ([yshift=-\outputheight cm]B.south);

        \end{tikzpicture}
        \vspace{-0.6em}
        \caption{\footnotesize Nonlocal game with inputs $(x, y)$ and outputs $(a, b)$}
        \end{subfigure}
        \vspace{-0.6em}
        \caption{\footnotesize Comparison between abstract systems and a nonlocal game interaction.}
    \end{figure}
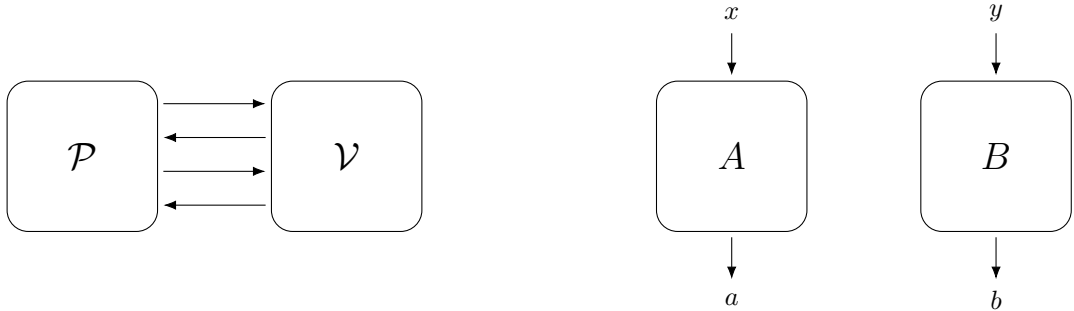
    
    It is quite clear that the two setups described so far---the standard nonlocal setup and the computational single-device setup---should be somehow connected.
    In the case of compiled nonlocal games, the connection is trivial:  one  can take a nonlocal game for the two-device setup and transform it to a single-$\qpt$-device setup using a computational assumption.
    But what is the fundamental link in the other direction (when we are not starting with a nonlocal game)?
    The computational assumptions regarding, e.g., claw-free functions are used as part of an interactive protocol between the classical verifier and the quantum prover, making it unclear how a computational assumption should be understood as some sort of a Bell inequality. 
    The novel work~\cite{kahanamoku2022classically} hinted at a connection by using the phrase ``computational Bell inequality'' and exploiting non-commuting measurements as in the CHSH Inequality, but a more fundamental and mathematically rigorous correspondence was not given.
    Thus, in order to gain a better fundamental understanding, taking protocols such as~\cite{kahanamoku2022classically,brakerski2023simple} for example, we ask:
    \begin{itemize}
        \item \emph{Can we pin down a ``computational Bell inequality''?}
        \item \emph{Can we formalize the link between computational assumptions and the nonsignaling assumption?}
    \end{itemize} 
    And, if so, 
    \begin{itemize}
        \item \emph{Can we use the well developed toolkit of nonlocal games to analyze the single-$\qpt$-device setup?}
    \end{itemize}
    
    We answer all these questions in the affirmative. We now briefly discuss our results.

\subsection*{Main results and ideas}
    
     We present a coherent and systematic approach for analyzing device-independent single-prover interactive protocols based on computational assumptions, through a clear link with Bell inequalities and nonlocal games.
    We expand upon the set of protocols addressed in prior studies~\cite{kahanamoku2022classically,brakerski2023simple,kalai2022compiled,compiledtrapdoor2024} by employing a canonical protocol and examine the interaction between a verifier and a prover:
            \begin{protocol}
            \vspace{-0.2em}
            \footnotesize
                \textbf{Canonical form protocols\footnote{For the complete form see Definition~\ref{def:canonical_form}.}}

                \vspace{-0.2em}


                \begin{enumerate}[label=Phase \Alph*., leftmargin=50pt]
                    \item
                    \begin{enumerate}[label=\arabic*.]
                        \item The verifier and prover interact classically and produce an interaction transcript~$\trans \in \Trans$.
                    \end{enumerate}
                    \item (Conditioned on passing Phase A)
                    \vspace{-0.6em}
                    \begin{enumerate}[label=\arabic*.]
                        \item The verifier samples a challenge~$\ib \leftarrow \Ib$ uniformly at random and sends it to the prover.
                        \item The prover responds with an output~$\ob \in \Ob$, which the verifier receives.
                    \end{enumerate}
                \end{enumerate}
                \vspace{-1.3em}
            \end{protocol}
            \vspace{-1.5em}

    \paragraph{Computational space of correlations.}
       
        In the standard Bell setting, analyzing the space of conditional probability distributions is highly insightful and fruitful.
        In particular, in the nonlocal space of correlations (SoC), the local set forms a convex polytope, with its facets corresponding to Bell inequalities.
        We define a computational analogue that we term the computational SoC (CSoC)-- this is done in Section~\ref{sec:com_soc}. 
        The CSoC that we construct by considering the correlations induced by the interaction in the canonical protocol underpins our approach.
        
        To derive the CSoC, we define a Bell mapping from the verifier-prover interaction to distributions over quadruples~$(\ia, \ib, \oa, \ob)$.
        A Bell mapping is a pair of functions~$({\imap:\Trans\rightarrow\Ia},{\omap:\Trans\rightarrow\Oa})$, with~$\Trans$ the transcript of Phase A of the protocol.
        Then, $(\ia=\imap(\trans), \oa=\omap(\trans))$ are what we call virtual variables (discussed more below) and~$(\ib, \ob)$ come from the real challenge and response in Phase B. 
        The set of the correlations over $(\ia, \ib, \oa, \ob)$ constructed this way gives the CSoC. 
        
        Note, however, that in the CSoC the input~$\ia=\imap(\trans)$ is \emph{not independent} from the two provers as in the usual nonlocal case. This is due to the shared dependency on the transcript, which is integral when starting from interactive protocol.
        We overcome this difficulty by working with the concept of \emph{measurement-dependent locality} (MDL), studied within the field of the foundations of quantum information~\cite{putz2014,putz2016mdl,scarani2019mdl}.

        Not every Bell mapping will induce a useful CSoC. As we discuss below, we require that the virtual input~$\ia=\imap(\trans)$ will be computationally hidden from the prover. This is where the computational assumptions of the protocol enter the picture. We expand on this in more detail below. 
        
        The combination of the above ideas allows us to define and study a computational local set~$\localcomp{\compLeakage}$, its MDL-extension~$\amdl{\kappa}$ and quantum set~$\quantumcomp{\compLeakage}$, with a respective leakage/signaling parameter~$\compLeakage$, which depends on the protocol and computational assumption. 

    \paragraph{Computational Bell inequalities.}
        Working with CSoC and MDL inequalities allows us to define new explicit Bell inequalities over the CSoC in their complete mathematical sense, hyperplanes that separate~$\localcomp{\compLeakage}$  and~$\amdl{\kappa}$ from~$\quantumcomp{\compLeakage}$.
        We prove that interactions with classical provers cannot violate our computational inequalities (this proof requires some effort; see Sections~\ref{sec:comp_local} and~\ref{sec:comp_bell}), while quantum ones can.
        Furthermore, as the leakage parameter~$\compLeakage$ increases, the ability to violate MDL-based inequalities---such as ours---is stronger compared to standard Bell inequalities, making such approach favorable.



    \paragraph{Hidden virtual inputs---where computational assumptions and locality meet.} 
        As mentioned above, the Bell map~$\imap$ and the transcript~$\trans$ define a virtual input~$\ia=\imap(\trans)$. 
        In the nonlocal setting, the two inputs to the two provers should be independent (or partially independent); This is the sense of locality. 
        Here as well one should enforce some structure or condition. 
        We say that the virtual input~$\ia$ is \emph{hidden} if for any QPT algorithm~$\A$ with polynomial advice, conditioned on passing the test of Phase A of the protocol, the following holds:
        \begin{equation}\label{eq:hidden}
            \underset{\trans}{\mathbbm{E}}
            \left|
            \Pr \big( \A(\state^\trans) = \imap(\trans) \big)
            - \frac{1}{|\Ia|}
            \right|
            \leq \compLeakage + \negl(\security) \; ,
        \end{equation}
        with~$\compLeakage\in[0,1]$ a leakage parameter, a security parameter~$\security$ and~$\state^\trans$ the prover's (unknown) quantum state resulting from the execution of the protocol.

        Equation~\eqref{eq:hidden} formalizes that, although the verifier can compute the virtual input~$\imap(\trans)$, a QPT prover can predict~$\imap(\trans)$ with only a limited advantage.
        In order for the CSoC and the computational Bell inequality to have the above discussed meaning, i.e., for them to be relevant for quantumness certification, one must \emph{prove} that Equation~\eqref{eq:hidden} holds for a given protocol.
        That is, the virtual input must be hidden.  
        To prove this, the computational assumption is employed (e.g., the trapdoor claw-free function or homomorphic encryption). 
        We thus clearly see how the computational assumption swaps with locality, by means of the virtual input.
        
        Remarkably, in our approach, this is the \emph{sole} place in which the computational assumption enters the analysis of the protocol.
        The rest of the analysis is completely oblivious to the assumption.
        This fact leads to modularity, removing many layers of technical overhead.  Furthermore, it highlights how one may go about constructing new protocols---as long as we can have a hidden virtual input, we are good to go. 

    \paragraph{Computational NPA-hierarchy.}
        
        
        Naturally, once we switch to working with the CSoC, we can start employing other tools from the study of non-locality.
        A leading example is the famous NPA-hierarchy~\cite{navascues2008npa,pironio2010npa2}.
        In Section~\ref{sec:comp_soc_hier}, we introduce a hierarchy of relaxations called \emph{computational-SoC hierarchy}, based on the  NPA-hierarchy, which approximates the correlations in~$\quantumcomp{\compLeakage}$, i.e., those achievable by efficient quantum provers. 
        Each level constrains signaling via measurements that reflect physically realizable strategies. 

        Although previous studies on compiled nonlocal games employed NPA-style hierarchies~\cite{klep2025seqnpa,cui2025seqnpa}, our method offers broader applicability (suitable for any protocol, not just compiled games, and allowing for $\compLeakage>0$) and greater simplicity (prior studies involved expectations of general noncommutative monomials in measurement operators~\cite{kulpe2024boundquantumvaluecompiled}, which might not represent physically feasible operations, thus requiring more technical steps).

       \paragraph{Analyzing single-prover protocols.} 

        Our observations are not only of fundamental nature but also have significant impact in terms of the ability to analyze mathematically the various certification protocols involving a classical verifier and a single prover.
        We use the protocol of~\cite{kahanamoku2022classically}, based on trapdoor claw-free function, and the one of~\cite{kalai2022compiled}, for compiled nonlocal games, as showcases for our method in Section~\ref{sec:showcases}.
        Clearly, the protocols are a priori very different.
        Nevertheless, we show how our techniques allow to analyze both of them rather easily and insightfully.

         We demonstrate the effectiveness of our approach by using the computational-SoC hierarchy to derive (a)~Tsirelson bounds for single-prover protocols---tighter than previous results and (b)~bound the entropy generated in the verifier-prover interaction---supplying a new tool and result, which can be further combined with our previous work on entropy accumulation in the single-prover setup to complete a randomness certification analysis~\cite{merkulov2023entropy}.
        The quantitative results are given in Section~\ref{sec:comp_soc_hier} (the interested reader may jump ahead to the plots in Figures~\ref{fig:kcvy_tsirelson} and~\ref{fig:entropy}).

    \subsection*{Previous and related works}
        
        Earlier proposals for quantum advantage rested on non-cryptographic, complexity-theoretic hardness of specific sampling tasks---most prominently boson sampling~\cite{aaronson2010boson,boixo2018nisqsupremacy,aaronson2023certified,bassirian2024fouriersmapling}. 
        By contrast, the breakthrough line on \emph{efficient} classical verification of quantum advantage with a single device~\cite{mahadev2018classical,brakerski2021cryptographic} leverages explicit post-quantum cryptographic primitives (notably LWE-based trapdoor claw-free functions and related tools) to achieve computational soundness with a tunable security parameter.

        Building on this perspective, additional classically verifiable quantum advantage tests were designed~\cite{kahanamoku2022classically,brakerski2023simple}.
        A complementary direction compiles any nonlocal game into a single-prover protocol while preserving quantum/nonlocal structure~\cite{kalai2022compiled,compiledtrapdoor2024}, with subsequent works initiating quantitative bounds on the compiled setting and studying convergence via sequential constraints~\cite{kulpe2024boundquantumvaluecompiled,cui2025seqnpa,klep2025seqnpa}.

         We highlight several recent works that are most closely related to our methodology and explain the main differences.

        \subsubsection*{Cryptographic single-device protocols (non-compiled)}
        
            \begin{itemize}
                \item \textbf{TCF-based test}~\cite{kahanamoku2022classically}.
                The protocol of~\cite{kahanamoku2022classically} falls within the family of protocols that our work considers.
                We instantiate our framework for this trapdoor claw-free-based test (see Section~\ref{sec:showcases}). 
                In addition, in Section~\ref{sec:ent_cert}, we prove that the protocol generates certified randomness against an unbounded adversary even when exposed to the transcript.
            
                \item \textbf{Simple tests of quantumness}~\cite{brakerski2023simple}. 
                This work studies tests are minimal representatives of protocols built directly from post-quantum primitives. 
                Our canonical-form protocol (Definition~\ref{def:canonical_form}) both subsumes \emph{and strictly generalizes} the protocol template of~\cite{brakerski2023simple}. 
                We believe that the analysis in our work is more insightful due to our ideas regarding the virtual hidden input and the computational SoC. 
                In terms of quantitative contributions, optimizing over our computational-SoC hierarchy yields leakage-dependent bounds; in the CHSH Bell scenario, our level-2 SDP gives a strictly tighter \emph{quantum} upper bound on the CHSH value than the analytic bound reported in~\cite[Theorem~5.2]{brakerski2023simple} (see Section~\ref{sec:cTs} and Fig.~\ref{fig:kcvy_tsirelson}).
                
            \end{itemize}

        \subsubsection*{Compiled nonlocal games}
        
        \begin{itemize}
            \item \textbf{Compiled nonlocal games}~\cite{kalai2022compiled,compiledtrapdoor2024}.
            The compiled-games paradigm starts from a Bell nonlocal game and compiles it into a single-prover protocol while preserving the game’s quantum/nonlocal structure, enabling a class of efficient verification of quantum advantage protocols. 
            We take the opposite, complementary, direction: from a general single-prover protocol in canonical form $\rightarrow$ a Bell inequality via a Bell mapping (see Section~\ref{sec:bell_map}).
            This reverse mapping exposes a virtual input hidden under a computational assumption and places all PPT strategies within a convex polytope, enabling \emph{computational} Bell inequalities that bound the computational–classical set.
            
            \item \textbf{Sequential games and sequential NPA hierarchy}\footnote{The term ``sequential'' in ``sequential games'' and in the ``sequential NPA hierarchy'' refers to distinct notions.}~\cite{kulpe2024boundquantumvaluecompiled,klep2025seqnpa,cui2025seqnpa}.
            Sequential games are introduced in~\cite{kulpe2024boundquantumvaluecompiled} to model step-by-step challenges in compiled nonlocal games and show that the optimal quantum value in the compiled game converges to that of the original game.
            \cite{klep2025seqnpa,cui2025seqnpa} develop a layered NPA hierarchy enforcing exact nonsignaling (up to negligible terms) on monomials of increasing length, yielding quantitative convergence rates.
            
            In contrast to these previous works, we enforce only approximate nonsignaling on physically realizable measurements, which suffices to capture practical prover strategies under leakage (see Section~\ref{sec:comp_soc_hier}), thus removing layers of technical overhead.
            Moreover, our canonical-form protocol generalizes sequential games~\cite{kulpe2024boundquantumvaluecompiled} and provides a unified framework that applies to arbitrary single-prover certification protocols.
            
            
        \end{itemize}

\section{Preliminaries}

    \subsection{Notation}
    
    We denote the indicator function as~$\indicator{\cdot}$.

    \begin{definition}[Total Variation Distance]
        Consider a measurable space~$(\Omega,\mathcal{F})$ and probability measures~$P$ and~$Q$, defined on~$(\Omega,\mathcal{F})$.
        The total variation distance between~$P$ and~$Q$ is defined as
        \begin{equation}
            \delta(P,Q)=\sup_{A\in\mathcal{F}}\qty|P(A)-Q(A)|\;.
        \end{equation}
    \end{definition}

    \subsection{Nonlocal games and Bell inequalities}
     
    Nonlocal games are mathematical constructs used to study quantum entanglement and nonlocality.
    They typically involve two players (or more) who are not allowed to communicate during the game.
    Each player receives an input, performs a local operation, and outputs a response.
    The distribution of inputs and outputs gives rise to observable correlations.

    In this work, we focus on two-player games, often specified by their input and output sets.
    A \emph{Bell scenario} is a tuple~$\bellscenarioShort = (\Ia, \Ib, \Oa, \Ob)$ describing:
    \begin{itemize}
        \item input sets~$\Ia$ and~$\Ib$ for players~$\playerA$ and~$\playerB$, respectively, and
        \item output sets~$\Oa$ and~$\Ob$ for~$\playerA$ and~$\playerB$, respectively.
    \end{itemize}
    As the structure of the relevant distribution sets only depends on the cardinalities~$(\bellscenarioInteger)$, we may refer to a Bell scenario either by its explicit sets or by their sizes~$(\bellscenarioInteger)$.
    
    We will typically assume that the inputs~$(\ia, \ib)$ are sampled from a fixed, known distribution---in most cases, the uniform distribution over~$\Ia \times \Ib$.
    However, we do not require these inputs to be independent of any hidden variables~$\hp$ used by the devices.
    This distinction is important: although standard Bell tests assume measurement independence, our framework accommodates input distributions that may be weakly correlated with the prover's internal state.
    This generalization is formalized later through the notion of measurement-dependent locality (MDL), and it plays a central role in our computational setting.

    We denote by~$\uc$ the set of all conditional distributions~$P(\oa, \ob \mid \ia, \ib)$ over a Bell scenario~$\bellscenarioShort$.
    The specific Bell scenario~$\bellscenarioShort$ will often be implicit from context.

    \begin{definition}[Nonsignaling set~$\ns$]\label{def:nonsignaling_set}
        Let~$\uc$ denote the set of all conditional probability distributions~$P(\oa, \ob \mid \ia, \ib)$ over~$\bellscenarioDistSet$.

        We say that a distribution~$P \in \uc$ belongs to the \emph{nonsignaling set}~$\ns$ if there exist:
        \begin{itemize}
            \item a hidden variable space~$\Gamma$ with a probability distribution~$\hpDensity$ over~$\Gamma$, and
            \item a family of conditional distributions~$\{P_\hp(\oa, \ob \mid \ia, \ib)\}_{\hp \in \Gamma}$,
        \end{itemize}
        such that:
        \begin{enumerate}[label=(\roman*)]
            \item For all~$(\ia, \ib, \oa, \ob)$, we have:
            \begin{equation}
                P(\oa, \ob \mid \ia, \ib) = \int \hpDensity(\hp) \cdot P_\hp(\oa, \ob \mid \ia, \ib) \, d\hp \;.
            \end{equation}

            \item For each~$\hp \in \Gamma$, the conditional distribution~$P_\hp$ satisfies the nonsignaling conditions:
            \begin{align}
                \sum_{\ob \in \Ob} P_\hp(\oa, \ob \mid \ia, \ib)
                &= \sum_{\ob \in \Ob} P_\hp(\oa, \ob \mid \ia, \ib') \quad \text{for all } \ib, \ib' \in \Ib \;, \label{eq:ns-a} \\
                \sum_{\oa \in \Oa} P_\hp(\oa, \ob \mid \ia, \ib)
                &= \sum_{\oa \in \Oa} P_\hp(\oa, \ob \mid \ia', \ib) \quad \text{for all } \ia, \ia' \in \Ia \;. \label{eq:ns-b}
            \end{align}
        \end{enumerate}
    \end{definition}

    \begin{definition}[Local set~$\local$]\label{def:local_set}
        A conditional distribution~$P \in \uc$ is said to belong to the \emph{local set}~$\local$ if there exist:
        \begin{itemize}
            \item a hidden variable space~$\Gamma$ with a probability distribution~$\hpDensity$ over~$\Gamma$, and
            \item a family of local response distributions~$\{ P_\hp(\oa | \ia),\; P_\hp(\ob | \ib) \}_{\hp \in \Gamma}$,
        \end{itemize}
        such that for all~$\ia, \ib, \oa, \ob$,
        \begin{equation}
            P(\oa, \ob | \ia, \ib)
            = \int d\hp~ \hpDensity(\hp) \cdot P_\hp(\oa | \ia) \cdot P_\hp(\ob | \ib) \;.
        \end{equation}
    \end{definition}

    \begin{definition}[Quantum set~$\quantum$]\label{def:quantum_set}
        A distribution~$P \in \uc$ is said to belong to the \emph{quantum set}~$\quantum$ if there exist:
        \begin{itemize}
            \item a finite-dimensional Hilbert space~$\mathcal{H}$,
            \item a normalized quantum state~$\rho$ on~$\mathcal{H}$,
            \item POVMs~$\{ M_\oa^\ia \}_{\oa \in \Oa}$ for each~$\ia \in \Ia$, acting on~$\mathcal{H}$,
            \item POVMs~$\{ N_\ob^\ib \}_{\ob \in \Ob}$ for each~$\ib \in \Ib$, acting on~$\mathcal{H}$,
        \end{itemize}
        such that for all~$\ia, \ib, \oa, \ob$,
        \begin{equation}
            P(\oa, \ob \mid \ia, \ib)
            = \tr\left( (M_\oa^\ia \otimes N_\ob^\ib) \rho \right) \;.
        \end{equation}
    \end{definition}

    \begin{definition}[Bell inequality]\label{def:bell_inequality}
        Let~$\uc$ denote the set of conditional distributions~$P(\oa, \ob \mid \ia, \ib)$ over a Bell scenario.
    
        We say that a function~$\bellInequality : \uc \to \mathbb{R}$ is a \emph{Bell inequality} if
        \begin{equation}\label{eq:bell_inequality_def}
            \bellInequality(P) \leq 0 \quad \text{for all } P \in \local \;.
        \end{equation}
    \end{definition}
    
    \newcommand{\bellinequalitySum}{\sum_{\oa\in\Oa}\sum_{\ob\in\Ob}\sum_{\ia\in\Ia}\sum_{\ib\in\Ib} }
    \begin{remark}
        In this work, we restrict attention to \emph{affine} Bell inequalities, meaning functionals of the form
        \begin{equation}
            \bellInequality(P) = \bellInequality_0 + \bellinequalitySum v_{\oa,\ob,\ia,\ib} P(\oa, \ob, \ia, \ib) \;,
        \end{equation}
        where~$v_{\oa,\ob,\ia,\ib} \in \mathbb{R}$ and~$\bellInequality_0 \in \mathbb{R}$ are fixed coefficients.
        This includes both tight and non-tight inequalities; the former correspond to facets of the local polytope~$\local$.
    \end{remark}

    \begin{remark}    
        In the literature, the term ``Bell inequality'' is sometimes reserved only for \emph{nontrivial} inequalities---those that are violated by at least one quantum or nonsignaling distribution~$P \notin \local$.
        For example, this viewpoint is adopted in~\cite{scarani2019formalizing}, where certain facets of the local polytope are explicitly excluded from the definition of Bell inequalities because they admit no quantum violation.
        By contrast, geometric and polyhedral approaches often refer to all valid affine constraints for the local set as Bell inequalities, whether or not they are violated by quantum distributions~\cite{pironio2014all,escola2021tight}.
        Our usage is inclusive: we refer to any such affine constraint satisfied by all~$P \in \local$ as a Bell inequality.
    \end{remark}

    Two examples:
    \begin{enumerate}
        \item \textbf{CHSH inequality.}
        In the Bell scenario~$\bellscenarioShort=(2,2,2,2)$, the CHSH inequality takes the form
        \begin{equation}
            \bellInequality(P) = -3/4 + \bellinequalitySum
            P(\oa,\ob,\ia,\ib)\cdot\indicator{\ia\cdot\ib = \oa\oplus\ob}\;.
        \end{equation}
        This is a nontrivial Bell inequality that is violated by some quantum correlations.
        
        \item \textbf{Trivial Bell inequality.}
        As a degenerate example, the constant functional
        \begin{equation}
            \bellInequality(P) = -\bellinequalitySum P(\oa,\ob,\ia,\ib) = -1
        \end{equation}
        satisfies the Bell inequality condition for all $P \in \uc$, but provides no nonlocality detection.
    \end{enumerate}

    \subsection{Measurement dependent locality (MDL)}\label{sec:mdl_pre}

    Standard Bell scenarios assume the principle of \emph{measurement independence}---also known as \emph{freedom of choice}---which posits that the inputs~$(\ia, \ib)$ are chosen independently of any hidden variables~$\hp$ used by the device.
    The \emph{measurement-dependent locality} (MDL) framework relaxes this assumption by allowing limited correlations between inputs and hidden variables.

    In MDL, the input distribution~$P(\ia, \ib \mid \hp)$ is constrained to lie within specified bounds, typically quantified by parameters~$(l,h)$.
    This defines the \emph{MDL set}, a relaxation of the local set that permits bounded measurement dependence.
    MDL models are useful in scenarios where input choices may be partially predictable or correlated with the device.
    They provide a structured way to analyze the robustness of Bell inequality violations under such constraints.

    \begin{definition}[MDL set~$\mdl{(l,h)}$]\label{def:mdl_set}
        Let~$\uc$ denote the set of joint distributions~$P(\oa, \ob, \ia, \ib)$ over a Bell scenario~$\bellscenarioShort = (\Ia, \Ib, \Oa, \Ob)$.

        Fix parameters~$0 \leq l \leq h \leq 1$.
        A distribution~$P \in \uc$ is said to belong to the \emph{measurement-dependent local set}~$\mdl{(l,h)}$ if there exist:
        \begin{itemize}
            \item a hidden variable space~$\Gamma$ with a probability distribution~$\hpDensity$ over~$\Gamma$,
            \item a family of local response distributions~$\{ P_\hp(\oa \mid \ia),\; P_\hp(\ob \mid \ib) \}_{\hp \in \Gamma}$, and
            \item a conditional input distribution~$P(\ia, \ib \mid \hp)$ satisfying the bounds
            \begin{equation}
                l \leq P(\ia, \ib \mid \hp) \leq h \quad \text{for all } \ia \in \Ia,\; \ib \in \Ib,\; \hp \in \Gamma \;,
            \end{equation}
        \end{itemize}
        such that the joint distribution factors as
        \begin{equation}
            P(\oa, \ob, \ia, \ib)
            = \int d\hp~ \hpDensity(\hp) \cdot P(\ia, \ib \mid \hp) \cdot P_\hp(\oa \mid \ia) \cdot P_\hp(\ob \mid \ib) \;.
        \end{equation}
    \end{definition}

    \begin{remark}
        An \emph{MDL inequality} is a Bell inequality (Definition~\ref{def:bell_inequality}) that holds for all distributions in the measurement-dependent local set~$\mdl{(l, h)}$.
        These inequalities generalize standard Bell inequalities to scenarios where the input distribution~$P(\ia, \ib)$ may be weakly correlated with a hidden variable~$\hp$.
        In this setting, MDL inequalities serve as linear constraints that distinguish classical strategies under bounded measurement dependence from more general behaviors, such as those achievable in quantum or general probabilistic theories.

        Just as tight Bell inequalities correspond to facets of the local polytope~$\local$, tight MDL inequalities define facets of the MDL polytope~$\mdl{(l,h)}$.\footnote{The MDL set~$\mdl{(l,h)}$ forms a convex polytope in the space of correlations~\cite[Theorem~2]{putz2014}.}
        They provide a natural extension of the Bell inequality framework to cases where full measurement independence does not hold.
    \end{remark}

    \begin{claim}[CHSH MDL inequality~{\cite[Equation (5)]{putz2014}}]
        \label{claim:putz2014} 
        Let $l,h\in[0,1]$ be MDL parameters for the Bell scenario~$\bellscenarioShort=(2,2,2,2)$.
        The following functional, is a nontrivial MDL inequality for any $l>0$.
        I.e., for any~$P\in\mdl{(l,h)}$,
        \newcommand{\dist}{P_{ABXY}}
        \begin{equation}\label{eq:gisin_inequality}
            \bellInequality(\dist) =
            l\dist(0000)-h\qty(\dist(0101)+\dist(1010)+\dist(0011)) \leq 0\;.
        \end{equation}
    \end{claim}

\section{Computational space of correlations}\label{sec:com_soc}

    In this section, we define the \emph{computational space of correlations}---a central technical object that underpins our framework.
    In the standard Bell setting, analyzing the space of conditional probability distributions reveals that the local set forms a convex polytope, with its facets corresponding to Bell inequalities.
    We define a computational analogue of this structure by considering the correlations induced by canonical verifier--prover protocols under computational constraints.
    
    By translating such interactions into a Bell scenario via a Bell mapping, we obtain distributions over tuples~$(\ia, \ib, \oa, \ob)$, where~$(\ia, \oa)$ are virtual variables extracted from the transcript and~$(\ib, \ob)$ come from the real challenge and response.
    We study how these distributions behave under classical and quantum strategies, when the prover is restricted to polynomial-time and the virtual input~$\ia$ is hidden.
    This leads to computational versions of the local and quantum sets, and sets the stage for constructing computational Bell inequalities that separate them.

    \subsection{Canonical protocol and Bell mapping}\label{sec:bell_map}

        Our work allows to use a native family of protocols, all instances of what we call the canonical protocol-- presented in Figure~\ref{def:canonical_form}. 
        The canonical protocol is a generalization of a class of protocols presented in~\cite{brakerski2023simple} and also covers the protocols from~\cite{kalai2022compiled,kahanamoku2022classically,compiledtrapdoor2024}.\footnote{We do remark that in the current form the canonical protocol does not fit, a priori at least, sampling based protocols such as~\cite{aaronson2010boson,boixo2018nisqsupremacy,aaronson2023certified,bassirian2024fouriersmapling}. We leave these type of protocol for future work.}  

        \begin{definition}[Canonical form protocol]\label{def:canonical_form}
            A verifier--prover interactive protocol is \emph{in canonical form} if and only if it follows the two–phase template shown in Figure~\ref{fig:canonical_form_protocol}.
        \end{definition}

        \begin{figure}[h!]
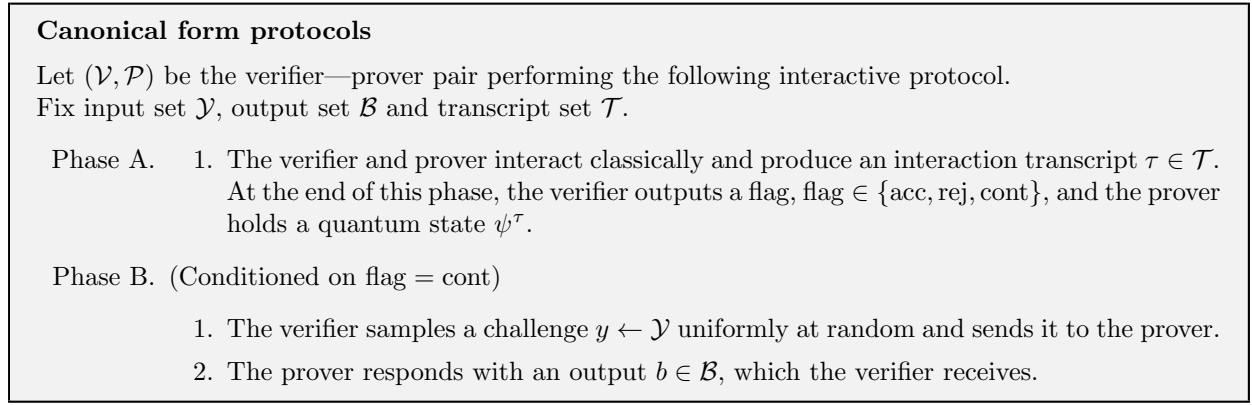

            \centering
            \begin{protocol}
                \textbf{Canonical form protocols}

                \vspace{0.5em}

                Let~$(\V,\P)$ be the verifier–-prover pair performing the following interactive protocol.

                Fix input set~$\Ib$, output set~$\Ob$ and transcript set~$\Trans$.
                \begin{enumerate}[label=Phase \Alph*., leftmargin=50pt]
                    \item
                    \begin{enumerate}[label=\arabic*.]
                        \item The verifier and prover interact classically and produce an interaction transcript~$\trans \in \Trans$.
                        At the end of this phase, the verifier outputs a flag,~$\flag \in \{\accept, \reject, \continue\}$, and the prover holds a quantum state~$\state^{\trans}$.
                    \end{enumerate}
                    \item (Conditioned on~$\flag = \continue$)
                    \begin{enumerate}[label=\arabic*.]
                        \item The verifier samples a challenge~$\ib \leftarrow \Ib$ uniformly at random and sends it to the prover.
                        \item The prover responds with an output~$\ob \in \Ob$, which the verifier receives.
                    \end{enumerate}
                \end{enumerate}
            \end{protocol}
            \vspace{-1em}
            \caption
            {
                \footnotesize
                Canonical form protocol structure.
            }
            \label{fig:canonical_form_protocol}
        \end{figure}   
            
        To make a precise link between the interactions of the verifier and the prover in the canonical protocol and a space of correlation (SoC), we define a ``Bell mapping''. Formally:
        \begin{definition}[Bell mapping]\label{def:bell_mapping}
            Let~$(\V,\P)$ be a pair of verifier-prover following  the canonical form protocol.
            Let~$(\Ib,\Ob,\Trans)$ be the inputs, outputs and transcripts sets (resp.).
            A Bell mapping of~$(\V,\P)$ to a Bell-Scenario~$\bellscenarioShort=(\bellscenario)$ is a pair of functions~$({\imap:\Trans\rightarrow \Ia},{\omap:\Trans\rightarrow \Oa})$.
        \end{definition}

        \begin{remark}
            The Bell mapping reinterprets the transcript of a canonical verifier--prover protocol as a virtual interaction in a Bell scenario.
            It extracts a synthetic input--output pair~$(\ia, \oa)$ from the transcript, which is then paired with the prover's real input--output pair~$(\ib, \ob)$ to form a quadruple~$(\ia, \ib, \oa, \ob)$.
            This allows the behavior of the protocol to be analyzed using tools from nonlocal games and Bell inequalities.
        \end{remark}




        \begin{definition}[Bell-mapped distribution]\label{def:bell_mapped_distribution}
            Let~$\security$ be a security parameter.
            Let~$(\V,\P)$ be a verifier-prover pair running a canonical-form protocol (Definition~\ref{def:canonical_form}), and let~$(\imap,\omap)$ be a Bell mapping (Definition~\ref{def:bell_mapping}).
            Let~$\mathsf{T}$ denote the (conditional) transcript distribution~$\Pr(\trans \mid \flag=\continue)$ induced by~$(\V,\P)$.
            
            The \emph{Bell-mapped distribution}~$P_\security$ is the joint distribution on~$\bellscenarioDistSet$ defined by
            \begin{equation}
              P_\security(\oa,\ob,\ia,\ib)
              \;\coloneqq\;
              \underset{\trans}{\E}\,
              \Pr\!\left(
                \omap(\trans)=\oa,
                \ob,
                \imap(\trans)=\ia,
                \ib
                \mid
                \trans
              \right).
            \end{equation}
            Here $\ib$ is the verifier’s Phase~B input sampled uniformly from~$\Ib$ and independently from~$\trans$, and $\ob$ is the prover’s Phase~B reply.
            \end{definition}

            The choice of the Bell mapping is, a priori, very flexible. However, we need to make a smart choice of the map for the analysis in the following sections. Specifically, we require a Bell mapping with a specific property. That is, the virtual input $\imap(\trans)$ defined via the map $\imap:\Trans\rightarrow \Ia$ should be computationally hidden. 
            Formally:
            

            \begin{definition}[Hidden virtual input~$\imap(\trans)$]\label{assumption:compLeakage}
                Let~$\compLeakage\in[0,1]$ be a leakage parameter and~$\security$ a security parameter.
                Let~$(\V,\P)$ be a canonical-form protocol, and let~$(\imap,\omap)$ be a Bell mapping.
                Let~$\state^\trans$ denote the prover’s post-interaction quantum state together with the transcript~$\trans$ (embedded as classical side information).
                We say the virtual input~$\ia=\imap(\trans)$ is \emph{hidden} if, for every QPT algorithm~$\A$ (possibly with polynomial advice),
                \begin{equation}
                    \underset{\trans}{\E}\!
                    \left|
                    \Pr\!\left( \A(\state^\trans)=\imap(\trans) \right)
                    -\frac{1}{|\Ia|}
                    \right|
                    \leq \compLeakage + \negl(\security)\;.
                \end{equation}
                The expectation is over the (conditional) distribution of transcripts produced by $(\V,\P)$ given $\flag=\continue$, and the probability is over the internal randomness of $\A$ (and any measurements it performs on $\state^\trans$).
            \end{definition}

            As mentioned, when analyzing a specific protocol it is important to choose a Bell mapping that will allow us to show that the virtual input~$\imap(\trans)$ is indeed hidden. 
            \emph{This is the formal connection between the original protocol's computational assumption and the nonsignaling assumption over the SoC.} 

            For compiled nonlocal games~\cite{kalai2022compiled,compiledtrapdoor2024}, the Bell mapping is straightforward, since the protocols were constructed from an underlying nonlocal game and Bell inequality.
            However, for other protocols, choosing a Bell mapping is more nuanced.
            We present Bell mappings that satisfy Definition~\ref{assumption:compLeakage} as showcases in Section~\ref{sec:showcases}.
    
            Note that while we exemplify the idea of the Bell mapping and the virtual input with known protocols, taking the other direction can be fruitful as well. 
            That is, one can try to come up with new protocols, or employ new cryptographic assumptions, by knowing which requirement they need to fulfill---having a Bell mapping that leads to a hidden virtual input. 
            This research avenue for finding new protocols can be seen as the parallel of looking for new Bell inequalities that can, e.g., certify more randomness or different entangled states. 
        
    \subsection{The computational classical set}\label{sec:comp_local}
        
        In this section, we investigate how the interaction between a classical verifier and a classical prover gives rise to a structured distribution over a Bell scenario via a Bell mapping.
        We formalize this idea by defining the \emph{computational classical set}---the set of Bell distributions that arise from classical strategies under a computational hiding constraint.

        \begin{definition}[Computational classical set]\label{def:comp_classical_set}
            Let~$\compLeakage\geq 0$.
            We say that a distribution~$C$ belongs to the \emph{computational classical set}~$\localcomp{\compLeakage}$ if there exists a classical verifier--prover pair~$(\V, \P)$ performing a canonical form protocol (Definition~\ref{def:canonical_form}), and a Bell mapping~$\bellmap$, such that:
            \begin{enumerate}
                \item The Bell mapping~$\bellmap$ satisfies the hidden input condition with leakage~$\compLeakage$ (Definition~\ref{assumption:compLeakage});
                \item Letting~$P_\security$ be the Bell-mapped distribution arising from the interaction between~$\V$ and~$\P$ at security parameter~$\security$ (Definition~\ref{def:bell_mapped_distribution}), we have
                \begin{equation}
                    \lim_{\security \to \infty} \delta(P_\security, C) = 0 \;.
                \end{equation}
            \end{enumerate}
        \end{definition}

        Ideally, we would like to interpret the resulting distributions as belonging to the standard local polytope~$\local$---that is, as a classical correlation within the SoC.
        However, we cannot directly apply the standard notion of locality, because the distribution induced by the interaction may exhibit dependence between the prover's behavior and the verifier's virtual input~$\ia$.
        This violates the usual assumption that inputs are chosen independently of any hidden variables used to generate the outputs.

        To address this, we turn to a more flexible and well-developed framework known as \emph{measurement-dependent locality (MDL)}~\cite{scarani2019mdl,putz2016mdl}, in which the input distribution~$P(\ia, \ib \mid \hp)$ is only required to lie within fixed bounds~$(l,h)$ (see Definition~\ref{def:mdl_set} in the Preliminaries).

        Our goal is to relate the interaction between the verifier and a classical prover to a distribution in the MDL set.
        However, this approach faces an obstacle: while the prover cannot predict the virtual input~$\ia$ with high accuracy on average (as enforced by the leakage bound~$\compLeakage$), there may still exist individual transcripts~$\trans$ in which~$\ia$ is fully determined.
        This rules out any pointwise guarantee of bounded dependence, and thus prevents us from directly mapping the interaction into the standard MDL set.

        To capture this more nuanced behavior, we define a model related to MDL that allows even greater flexibility in the dependence between~$\ia$ and the hidden variable, while assuming~$\ib$ remains independent.
        We call this the \emph{one-sided average measurement-dependent local set}, or AMDL for short, and denote it by~$\amdl{\kappa}$.
        
        \begin{definition}[One-sided average measurement-dependent local set~$\amdl{\kappa}$ (AMDL)]\label{def:amdl_set}
            Let~$\kappa \geq 0$.

            A distribution~$P(a,b,x,y)$ is said to belong to the \emph{one-sided average measurement-dependent local set}~$\amdl{\kappa}$ if there exist:
            \begin{itemize}
                \item a hidden variable space~$\Gamma$ with a probability distribution~$\hpDensity$ over~$\Gamma$,

                \item a family of local conditional output distributions~$\{ P_\hp(\oa \mid \ia),\; P_\hp(\ob \mid \ib) \}_{\hp \in \Gamma}$, and
                \item a conditional input distribution~$P(\ia \mid \hp)$ and a uniform~$P(\ib)$,
            \end{itemize}
            such that:
            \begin{enumerate}[label=(\roman*)]
                \item the joint distribution is given by
                \begin{equation}\label{eq:amdl_factorization}
                    P(\oa, \ob, \ia, \ib)
                    = \int d\hp~ \hpDensity(\hp) \cdot P(\ia \mid \hp) \cdot P(\ib) \cdot P_\hp(\oa \mid \ia) \cdot P_\hp(\ob \mid \ib) \;,
                \end{equation}
                    \item and the expected deviation of the most likely input from uniform is bounded by~$\kappa$:
                    \begin{equation}\label{eq:amdl_constraint}
                        \E_\hp \left[
                            \max_{\ia} P(\ia \mid \hp) - \frac{1}{|\Ia|}
                        \right] \leq \kappa \;.
                    \end{equation}
            \end{enumerate}
        \end{definition}

        In our context, the parameter~$\kappa$ corresponds to the guessing advantage that a classical prover may have on the virtual input~$\imap(\trans)=\ia$, as quantified by the leakage bound~$\compLeakage$.
        The following lemma shows that this guessing bound implies that the Bell-mapped distribution induced by any classical prover belongs to~$\amdl{\compLeakage}$.

        \begin{lemma}\label{lemma:classical_amdl_reduction}
            Let $\compLeakage\in[0,1]$.
            Then the local computational set $\localcomp{\compLeakage}$ is a subset of the closure of the one-sided average measurement-dependent local set $\amdl{\compLeakage}$.
        \end{lemma}

        \begin{proof}
            Let~$C \in \localcomp{\compLeakage}$.
            Then, by Definition~\ref{def:comp_classical_set}, there exists a classical verifier--prover pair~$(\V, \P)$ performing a canonical form protocol and a Bell mapping~$\bellmap$ satisfying the hidden input condition with leakage~$\compLeakage$, such that the Bell-mapped distribution~$P_\security$ converges to~$C$ in variation distance.

            We now construct a distribution~$\plrc$ that both:
            \begin{enumerate}[label=(\roman*)]
                \item reproduces the same statistics as~$P_\security$, and
                \item $\plrc\in\amdl{(\compLeakage+\negl(\security))}$.
            \end{enumerate}

            We define~$\plrc$ by constructing a standard nonlocal game.
            The construction simulates the verifier-prover interaction and specifies how the hidden variable~$\hp$ is sampled, how the inputs~$(\ia, \ib)$ are chosen, and how each player responds based on their input and the shared hidden parameter.
            The procedure is illustrated in Figure~\ref{fig:comp_mdl} and works as follows. 

            \begin{figure}[t!]
                \centering
    
                \begin{subfigure}[t]{0.47\textwidth}
                \centering
                \begin{tikzpicture}[>=Latex, node distance=1.2cm and 1.6cm]
    
                \node (T)        at (2,2.5)       {$T$};
                \node (X)        at (0,2.5)       {$X$};
                \node (Y)        at (4,2.5)       {$Y$};
                \node[draw, rounded corners=6pt, minimum size=1cm] (PA) at (0,1) {$P_{A|XT}$};
                \node[draw, rounded corners=6pt, minimum size=1cm] (PB) at (4,1) {$P_{B|YT}$};
                \node[below=of PA, anchor=south] (A) {$A$};
                \node[below=of PB, anchor=south] (B) {$B$};
    
                \draw[->] (X) -- (PA);
                \draw[->] (T) -- (PA);
                \draw[->] (T) -- (PB);
                \draw[->] (Y) -- (PB);
                \draw[->] (PA) -- (A);
                \draw[->] (PB) -- (B);
                \draw[->] (T) -- (X);
    
                \end{tikzpicture}
    
                \vspace{0.5em}
                \caption
                {
                    \footnotesize Real MDL interpretation -- The ``hidden'' parameter~$T$, representing the transcript~$\trans$, decides the states both parties are holding and the value of~$\ia=\imap(\trans)$.
                }
                \label{fig:comp_mdl_i}
                \end{subfigure}
                \hfill
                \begin{subfigure}[t]{0.47\textwidth}
                \centering
                \begin{tikzpicture}[>=Latex, node distance=1.2cm and 1.6cm]
    
                \definecolor{verylightgray}{gray}{0.7} 
                \node (T) at (2,2.5) {\textcolor{verylightgray}{$T$}};
                \node (X)        at (0,2.5)       {$X$};
                \node (Y)        at (4,2.5)       {$Y$};
                \node[draw, rounded corners=6pt, minimum size=1cm] (PA) at (0,1) {$P_{A|X\Gamma}$};
                \node[draw, rounded corners=6pt, minimum size=1cm] (PB) at (4,1) {$P_{B|Y\Gamma}$};
                \node (G)        at (2,1.6)       {$\Gamma$};
                \node[below=of PA, anchor=south] (A) {$A$};
                \node[below=of PB, anchor=south] (B) {$B$};
    
                \draw[->, verylightgray] (T) -- (G);
                \draw[->, verylightgray] (T) -- (X);
                \draw[->] (X) -- (PA);
                \draw[->] (Y) -- (PB);
                \draw[->] (G) -- (PA);
                \draw[->] (G) -- (PB);
                \draw[->] (PA) -- (A);
                \draw[->] (PB) -- (B);
    
                \end{tikzpicture}
    
                \vspace{0.5em}
                \caption
                {
                    \footnotesize ``Processed'' interpretation -- The ``hidden'' parameter is~$\Gamma$, representing instructions for Player~$\playerB$.
                }
                \label{fig:comp_mdl_ii}
                \end{subfigure}
                \vspace{-0.6em}
                \caption
                {
                    \footnotesize MDL interpretation of the protocol template -- The verifier~$\mathcal{V}$ and device~$\mathcal{D}$ receive respective inputs~$x,m$ and respectively output~$a,b$.
                }
                \label{fig:comp_mdl}
            \end{figure}

            \begin{enumerate}
                \item \textbf{Sampling the hidden variable.}  
                The referee simulates an interaction between~$\V$ and~$\P$ to generate a transcript~$\trans$.
                Then, for every challenge~$\ib \in \Ib$, the referee rewinds~$\P$ and queries it on~$\ib$ using the same transcript~$\trans$, recording the response~$\ob_{\ib}$.
                Since~$\P$ is a classical PPT device, rewinding is allowed to extract consistent answers.
                Define the hidden parameter~$\hp := (\ob_{\ib})_{\ib \in \Ib} \in \Ob^{|\Ib|}$.

                \item \textbf{Input sampling.}  
                Let~$\ia := \imap(\trans)$ and sample~$\ib$ independently and uniformly at random from~$\Ib$.

                \item \textbf{Player strategies.}
                \begin{itemize}
                    \item \textbf{Player A}, on input~$\ia$ and shared parameter~$\hp$, samples a new transcript~$\trans'$ from the conditional distribution~$P(\trans' \mid \ia, \hp)$, defined to match the distribution over transcripts conditioned on~$\imap(\trans') = \ia$ and~$\hp$.
                    Then outputs~$\oa := \omap(\trans')$.
                    \item \textbf{Player B}, on input~$\ib$ and~$\hp$, returns~$\ob := \hp_{\ib}$.
                \end{itemize}
            \end{enumerate}

            \textbf{Claim (i).~$\plrc$ reproduces the Bell-mapped distribution.}
            We show that the resulting distribution~$\plrc$ is statistically indistinguishable from the Bell-mapped distribution~$P_\security$ obtained by applying the Bell mapping to the original verifier-prover interaction.

            If Player~$\playerA$ were given the original transcript~$\trans$ directly, and simply returned~$\oa := \omap(\trans)$, the resulting distribution would trivially match~$P_\security$ by definition.
            In our construction, however, Player~$\playerA$ samples~$\trans'$ conditioned on~$\imap(\trans') = \ia$ and~$\hp$.

            Let~$T$ denote the random variable representing the original transcript, and~$T'$ the one sampled by Player~$\playerA$.
            Let~$\Hp$ denote the hidden parameter.
            We claim that
            \begin{equation}\label{eq:statistical_equality_current}
                (T', \Hp, \imap(T)) \overset{d}{=} (T, \Hp, \imap(T)) \;.
            \end{equation}
            To see this, we expand the joint distribution:
            \begin{align}
                \Pr(T' = \tau, \Hp = \hp, \imap(T) = \ia)
                &= \Pr(T' = \tau \mid \Hp = \hp, \imap(T) = \ia) \cdot \Pr(\Hp = \hp, \imap(T) = \ia) \\
                &= \Pr(T = \tau \mid \Hp = \hp, \imap(T) = \ia) \cdot \Pr(\Hp = \hp, \imap(T) = \ia) \\
                &= \Pr(T = \tau, \Hp = \hp, \imap(T) = \ia) \;,
            \end{align}
            where the second equality holds by the definition of~$T'$ as sampling from the same conditional distribution as~$T$.

            Because the joint distributions over~$(T, \Hp, \imap(T))$ and~$(T', \Hp, \imap(T))$ are equal, the full joint distribution including~$\ib$ and deterministic functions of the transcript is preserved:
            \begin{align}
                (T', \Hp, \imap(T), \ib) &\overset{d}{=} (T, \Hp, \imap(T), \ib) \;, \\
                (\omap(T'), \Hp_\ib) &\overset{d}{=} (\omap(T), \Hp_\ib) \;.
            \end{align}
            Hence, the output tuple~$(\ia, \ib, \oa, \ob)$ under~$\plrc$ is identically distributed to~$P_\security$.

            \textbf{Claim (ii).~$\plrc\in\amdl{(\compLeakage+\negl(\security))}$.}
            Assume toward contradiction that this is not the case.
            Then there exists a non-negligible function~$\mu(\cdot)$ such that
            \begin{equation}\label{eq:contradiction_assumption}
                \E_\hp \left[
                    \max_{\ia} P(\ia \mid \hp) - \frac{1}{|\Ia|}
                \right] > \compLeakage + \mu(\security) \;.
            \end{equation}

            \newcommand{\adversary}{\mathcal{W}}
            We now construct a QPT adversary~$\adversary$ that breaks the hidden input assumption (Definition~\ref{assumption:compLeakage}).
            On input~$\state^{\trans}$, which includes the transcript~$\trans$, the adversary computes~$\hp := (\ob_{\ib})_{\ib \in \Ib}$ by rewinding the prover~$\P$ on all inputs~$\ib$ (which is possible since~$\P$ is classical and polynomial-time).
            It then uses advice~$z_\hp \in \Ia$ corresponding to the most likely input~$\ia$ under~$P(\ia \mid \hp)$.
            That is,
            \begin{equation}
                z_\hp \coloneqq \arg\max_{\ia'} P(\ia' \mid \hp) \;.
            \end{equation}

            Note that the space of hidden parameters~$\hp$ is constant (bounded by~$|\Ob|^{|\Ib|}$), so the advice table~$\hp \mapsto z_\hp$ has constant size with respect to the security parameter~$\security$.
            Thus~$\adversary$ is a QPT adversary with constant advice.
            The success probability of~$\adversary$ is then
            \begin{equation}
                \Pr_{\trans}[\adversary(\state^\trans) = \imap(\trans)] = \E_\hp \left[ \max_\ia P(\ia \mid \hp) \right] \;.
            \end{equation}
            Combining this with Equation~\eqref{eq:contradiction_assumption}, we have
            \begin{equation}
                \Pr[\adversary(\state^\trans) = \imap(\trans)] > \frac{1}{|\Ia|} + \compLeakage + \mu(\security) \;,
            \end{equation}
            which contradicts the hiding assumption (Definition~\ref{assumption:compLeakage}), since~$\mu(\security)$ is non-negligible.
            
            Therefore,~$P_\security \in \amdl{\compLeakage + \negl(\security)}$.
            By the continuity of the AMDL set (Lemma~\ref{lemma:mdl_continuity}), it follows that~$C$ belongs to the closure of~$\amdl{\compLeakage}$.
        \end{proof}
    
        Having established that classical provers induce distributions in~$\amdl{\compLeakage}$, 
        we now study the structure of the set~$\amdl{\kappa}$ itself. 
        In particular, we would like to understand how a distribution in this set compares to the well-studied MDL sets. 
        The next lemma shows that any distribution in~$\amdl{\kappa}$ can be expressed as a convex combination of an MDL distribution (with slightly relaxed parameters) and an unconstrained remainder. 
        This decomposition will later allow us to translate guarantees for MDL inequalities into corresponding bounds for~$\amdl{\kappa}$.

    \subsection{From measurement dependent locality to computational Bell inequalities}\label{sec:comp_bell}

        This subsection is devoted to proving Theorem~\ref{thm:cl_to_local}.
        Operationally, the theorem furnishes an explicit \emph{computational Bell inequality} tailored to canonical-form protocols: after applying a Bell mapping, every classical (PPT) prover produces correlations that satisfy the inequality.
        Thus, we obtain a protocol-specific bound that no efficient classical strategy can surpass, while leaving room for quantum strategies to violate.

        \begin{figure}[t!]
            \centering
            \includegraphics[scale=0.3]{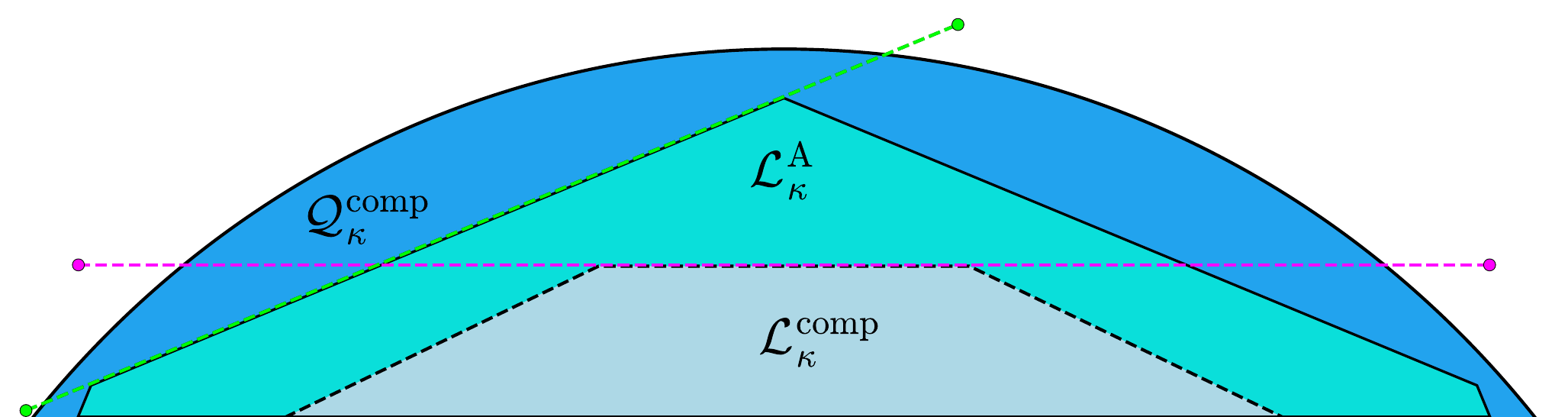}
            \caption{
                \footnotesize                
                A schematic~$2$-dimensional slice of correlation space for fixed leakage~$\compLeakage$ in the CHSH Bell scenario.
                The teal polygon is the average measurement-dependent local (AMDL) polytope~$\amdl{\compLeakage}$, which contains the computational-local set~$\localcomp{\compLeakage}$.
                The location of the computational-quantum set~$\quantumcomp{\compLeakage}$ relative to~$\amdl{\compLeakage}$ is not asserted; however, the two are disjoint, as a computational Bell inequality separates~$\quantumcomp{\compLeakage}$ from~$\amdl{\compLeakage}$ (and hence from~$\localcomp{\compLeakage}$).
                The dashed magenta line illustrates a shifted-CHSH inequality adapted to input leakage.
                The dashed lime line illustrates a facet-defining AMDL inequality, violated by some~$Q\in\quantumcomp{\compLeakage}$.
            }
            \label{fig:polytopes}
        \end{figure}

        \begin{theorem}\label{thm:cl_to_local}
            Let~$\compLeakage\in[0,1]$ and let~$\flexibility>0$.
            There exists an explicit computational Bell inequality $\I$ (with parameters depending on $\compLeakage$ and $\flexibility$) such that for any distribution~$\pcomp\in\localcomp{\compLeakage}$,
            \begin{equation}
                \I(\pcomp)\leq 0 \; .
            \end{equation}
        \end{theorem}

        \begin{proof}
            See Appendix~\ref{ap_sec:supp_material}.
        \end{proof}

        To prove Theorem~\ref{thm:cl_to_local} we rely on two lemmas. 
        The first, Lemma~\ref{lemma:amdl_mdl_decomposition_one_sided}, shows a one-sided decomposition for behaviors in $\amdl{\compLeakage}$: any such behavior can be written as a convex combination of a distribution inside a slightly relaxed MDL set $L\in\mdl{(l_{\flexibility},h_{\flexibility})}$ and an unconstrained distribution $S\in\uc$, with the weight on the unconstrained part controlled by $\compLeakage$ and the slack parameter $\flexibility$. 
        The second, Lemma~\ref{lemma:amdl_mdl_inequality_bound}, turns this structural statement into a bound for inequalities: starting from any MDL inequality valid for $\mdl{(l_{\flexibility},h_{\flexibility})}$, affinity implies its contribution on the MDL component is nonpositive, so it suffices to control the small unconstrained fraction~$S\in\uc$. 
        Lemma~\ref{lemma:amdl_mdl_inequality_bound} formalizes this transfer and yields an explicit loss that scales like $\compLeakage/(\compLeakage+\flexibility)$ against the inequality’s maximum over $\uc$.
        
        \begin{lemma}[One-sided MDL decomposition under AMDL, with explicit $l$]\label{lemma:amdl_mdl_decomposition_one_sided}
            Let~$\kappa \geq 0$ and~$\vartheta > 0$, and let~$P \in \amdl{\kappa}$.
            Then, one can always write
            \begin{equation}
                P \;=\; \left(1-\frac{\kappa}{\kappa+\vartheta}\right)\,L \;+\; \frac{\kappa}{\kappa+\vartheta}\,S \; ,
            \end{equation}
            with~$S \in \uc$ and~$L \in \mdl{(\mdlparams)}$ \textup{(see Definition~\ref{def:mdl_set})},
            \begin{equation}
                \mdlparamsExplicit \; .
            \end{equation}
        \end{lemma}

        \begin{lemma}[Bounding AMDL distributions by MDL inequalities]\label{lemma:amdl_mdl_inequality_bound}
            Let $\flexibility>0$.
            Define~$\local_\flexibility\coloneqq\mdl{(l_\flexibility, h_\flexibility)}$
            Where
            \begin{equation}
                \mdlparamsExplicit \;.
            \end{equation}
            Let $\bellInequality_\flexibility$ be any MDL inequality valid for $\local_\flexibility$.
            Then for any $P\in\amdl{\kappa}$,
            \begin{equation}
                \bellInequality_\flexibility(P) \;\le\; \frac{\kappa}{\kappa+\flexibility}\cdot
                \max_{S\in\uc}\;\bellInequality_\flexibility(S)
                \;=\; \O\!\left(\frac{\kappa}{\kappa+\flexibility}\right) \; .
            \end{equation}
        \end{lemma}

        \begin{proof}[Proofs of Lemmas~\ref{lemma:amdl_mdl_decomposition_one_sided} and~\ref{lemma:amdl_mdl_inequality_bound}]
            See Appendix~\ref{ap_sec:supp_material}.
        \end{proof}

        Let us explain the importance of Theorem~\ref{thm:cl_to_local} and compare it to previous works. 
        Firstly, the theorem allows us to reduce the problem of analyzing the interaction of~$(\V, \P)$ in the case of a PPT prover to the setup of local distributions. This means that all limitations that hold for local distributions in the nonlocal setup can be directly transferred to limitations on PPT provers. For example, there is no longer a need to analyze the winning probability of a PPT prover in a test of quantumness using proofs tailored to a specific protocol, as in~\cite{kahanamoku2022classically} for example.

        Secondly---and crucially---is the ability to tailor a Bell inequality to the relevant scenario.
        In our case this corresponds to the transition to an MDL inequality.
        MDL inequalities are stronger than, e.g., the CHSH inequality, when the (virtual) inputs are not fully independent from the behavior of the devices (the strategy of the prover). 
        In the case of certification using a single device using a computational assumption, this dependency comes in two forms: (1) Both the virtual input and the strategy of the prover depend on the transcript (2) The usage of the computational assumption does not lead to a completely hidden virtual input. 

        We can examine the consequence of this dependency in terms of certification using Figure~\ref{fig:polytopes}.
        When there is no dependency at all, one can use the CHSH inequality, which separates the local set $\local$ from the quantum one $\quantum$ (recall Figure~\ref{fig:sets}).
        Now, what is typically done when some dependency between the inputs and the strategy of the devices is introduced is to simply ``shift'' the CHSH inequality~\cite{brakerski2023simple,kahanamoku2022classically}; the shifted CHSH is denoted by the magenta dashed line in Figure~\ref{fig:polytopes}.
        In this situation, it is also harder for quantum correlations to violate the shifted inequality.
        At some point, the virtual hiding becomes so large that no distribution in $\mathcal{Q}$ can violate the shifted inequality.
        The MDL inequality, denoted by the lime dashed line, is tilted in a way that will always allow for some distribution in $\mathcal{Q}$ to violate it. Thus, the MDL inequalities are better suited for studying the correlations that arise from the canonical protocols.
        
        While prior work such as~\cite{brakerski2023simple} provide analytic bounds on the achievable CHSH score by quantum provers in the presence of input leakage, it is often unclear how to realize such winning strategies in concrete protocols.
        In particular, when a quantum prover attempts to gain information about the virtual input~$\ia$, it may be forced to perform a measurement or otherwise disturb its internal state.
        This can interfere with its ability to maintain a coherent superposition, which is essential for achieving quantum advantage in the CHSH game.

        In effect, the tasks of \emph{guessing the virtual input} and \emph{winning the game} may conflict.
        As a result, even though the theoretical upper bound for quantum violation increases with leakage, it may not be achievable in practice within a given protocol.
        This tension motivates the use of MDL inequalities: unlike the shifted CHSH inequality, an MDL inequality is structurally adapted to the presence of input-strategy correlation, and can better account for this tradeoff.

        To summarize, we have shown that any classical prover, when mapped into a Bell scenario using a leakage-bounded Bell mapping, induces a distribution that lies within a one-sided measurement-dependent local set~$\amdl{\kappa}$.
        By combining this observation with an MDL inequality tailored to~$\kappa$, we obtain a bound on the degree to which any classical prover can violate the inequality.
        This provides a computationally meaningful analogue of classical locality in the space of correlations, and sets the stage for understanding which behaviors remain possible when the prover is quantum.

        \subsection{The computational quantum set}\label{sec:comp_quantum}

        In the previous section, we analyzed the local distributions arising from interactions between a verifier and a classical PPT prover in the canonical protocol.
        We showed that these distributions belong to the set~$\amdl{\compLeakage}$ and satisfy a corresponding MDL inequality, thereby establishing computational soundness for classical strategies.

        To use this framework for the certification of quantum provers, we must now extend the analysis to quantum interactions.
        That is, we need to characterize the correlations and internal states generated when the prover is an efficient quantum device (i.e., QPT).
        This requires a description of the structure of the quantum states used or generated by the prover in the protocol and how they relate to the Bell-mapped inputs and outputs.

        We begin with a definition of the quantum set.

        \begin{definition}[Computational quantum set]\label{def:comp_quantum_set}
            We say that a distribution~$Q$ belongs to the \emph{computational quantum set}~$\quantumcomp{\compLeakage}$ if there exists a verifier-prover pair~$(\V, \P)$, where~$\P$ is a QPT device, performing a canonical form protocol (Definition~\ref{def:canonical_form}), and a Bell mapping~$(\imap, \omap)$, such that:
            \begin{enumerate}
                \item The Bell mapping~$(\imap, \omap)$ satisfies the hidden input condition with leakage~$\compLeakage$ (Definition~\ref{assumption:compLeakage});
                \item Letting~$P_\security$ be the Bell-mapped distribution arising from the interaction between~$\V$ and a QPT prover~$\P$ at security parameter~$\security$ (Definition~\ref{def:bell_mapped_distribution}), we have
                \begin{equation}
                    \lim_{\security \to \infty} \delta(Q, P_\security) = 0 \;.
                \end{equation}
            \end{enumerate}
        \end{definition}

        In the quantum computational setting, a canonical protocol induces a distribution over transcripts $\trans$, and each transcript determines the prover's (possibly subnormalized) final state $\state^{\trans}$.
        By grouping these states according to the Bell-mapped input $\ia=\imap(\trans)$ and output $\oa=\omap(\trans)$, we obtain a family of virtual states $\{\state^{\oa\mid\ia}\}_{\ia,\oa}$ together with their input marginals $\{\state^{\ia}\}_{\ia}$.
        Intuitively, $\state^{\oa\mid\ia}$ is the prover's final state conditioned on the input--output pair, while $\state^{\ia}$ averages over outputs at the same input.
        This representation will be convenient for the mathematical analysis of QPT provers' capabilities within our framework.

        \begin{definition}[Input and input--output conditioned states]\label{def:virtual_state_family}
            Let~$(\V, \P)$ be a verifier--prover pair performing a canonical form protocol with transcript set~$\Trans$ and Bell mapping~$\bellmap$.
            For each~$\ia \in \Ia$ and~$\oa \in \Oa$, define the subnormalized state~$\state^{\oa \mid \ia}$ as
            \begin{equation}
                \state^{\oa \mid \ia} \coloneqq
                \sum_{\trans \in \Trans \,:\, \imap(\trans) = \ia,\; \omap(\trans) = \oa}
                \Pr(\trans \mid \imap(\trans) = \ia) \cdot \state^{\trans} \; .
            \end{equation}
        
            The (normalized) input-conditioned states~$\state^{\ia}$ are then defined as
            \begin{equation}
                \state^{\ia} \coloneqq
                \sum_{\oa \in \Oa} \state^{\oa \mid \ia} \; .
            \end{equation}
        \end{definition}

        This construction naturally leads to a quantum correlation over the Bell scenario defined by the mapping~$(\imap, \omap)$.
        Indeed, once the prover's state is conditioned on a particular virtual input~$\ia$, the canonical protocol specifies how the prover processes a challenge~$\ib \in \Ib$ by applying the POVM~$\{\measby\}_{\ob\in\Ob}$ to the state~$\state^{\ia}$.
        The outcome~$\ob \in \Ob$ completes the correlation tuple~$(\ia, \ib, \oa, \ob)$.

        More precisely, the probability of observing outcomes~$(\oa, \ob)$ given inputs~$(\ia, \ib)$ is determined by:
        \begin{equation}
            P(\oa, \ob \mid \ia, \ib) = \tr\left[ \state^{\oa\mid\ia} \measby \right] \;.
        \end{equation}
        This defines a quantum correlation over the Bell scenario~$(\Ia, \Ib, \Oa, \Ob)$, capturing both the structure of the canonical protocol via~$\state^{\oa \mid \ia}$, and the quantum prover's measurement strategy in Phase~B via~$\measby$.

        \begin{lemma}[State-based representation of $\quantumcomp{\compLeakage}$]\label{lem:state_representation}
            Fix a canonical-form protocol and Bell mapping $(\imap,\omap)$, and let
            $\pi(\ia,\ib)$ denote the verifier's joint input distribution over $\Ia\times\Ib$
            (not necessarily uniform or independent).
            Let $\P$ be a QPT prover and let $\{\state^{\oa\mid\ia}\}_{\ia,\oa}$ and
            $\{\measby\}_{\ib,\ob}$ be as in Definition~\ref{def:virtual_state_family}.
            Then, for all $\ia,\ib,\oa,\ob$,
            \begin{equation}
              P(\oa,\ob \mid \ia,\ib) \;=\; \tr\!\left[\,\state^{\oa\mid\ia}\,\measby\,\right]
              \qquad\text{and}\qquad
              P(\ia,\ib,\oa,\ob) \;=\; \pi(\ia,\ib)\,\tr\!\left[\,\state^{\oa\mid\ia}\,\measby\,\right] \; .
            \end{equation}
        \end{lemma}

        \begin{proof}
            By definition,
            \begin{equation}
              \state^{\oa\mid\ia}
              =
              \sum_{\tau:\,\imap(\tau)=\ia,\;\omap(\tau)=\oa}
              \Pr(\tau \mid \imap(\tau)=\ia)\,\state^{\tau}\,.
            \end{equation}
            Hence, for any $\ib$ and $\ob$,
            \begin{equation}
              P(\oa,\ob \mid \ia,\ib)
              =
              \sum_{\tau:\,\imap(\tau)=\ia,\;\omap(\tau)=\oa}
              \Pr(\tau \mid \imap(\tau)=\ia)\,
              \tr\!\left[ \state^{\tau}\,\measby \right]
              =
              \tr\!\left[ \state^{\oa\mid\ia}\,\measby \right]\,.
            \end{equation}
            If $\pi(\ia,\ib)$ denotes the (arbitrary) input distribution, then
            \begin{equation}
              P(\ia,\ib,\oa,\ob)
              =
              \pi(\ia,\ib)\,\tr\!\left[ \state^{\oa\mid\ia}\,\measby \right]\,.
            \end{equation}
        \end{proof}

        \begin{lemma}[Convexity of the computational quantum set]
            \label{lemma:computational_quantum_convexity}
            Let~$P_0$ and~$P_1$ be Bell-mapped distributions induced by two QPT provers interacting with a canonical verifier, and let~$q \in [0,1]$ be an efficiently computable real number.
            Then, for each~$\security \in \N$, there exists a QPT prover~$\tilde{\P}$ such that the Bell-mapped distribution~$\tilde{P}_\security$ induced by the interaction of~$\tilde{\P}$ with the verifier satisfies
            \begin{equation}
                \delta\left( \tilde{P}_\security,\; q P_0 + (1 - q) P_1 \right) \leq \negl(\security) \;.
            \end{equation}
        \end{lemma}

        \begin{proof}
            Let~$\security\in\N$ and let~$P_0$ and~$P_1$ be the Bell-mapped distributions induced at security~$\security$ by two QPT provers~$\P_0$ and~$\P_1$, respectively.
            Since~$q$ is efficiently computable, there exists a deterministic polynomial-time algorithm that, on input $1^\security$, outputs a rational $q_\security\in[0,1]\cap\mathbb{Q}$ with
            \begin{equation}
              \left|q_\security - q\right| \leq \negl(\security) \; .
            \end{equation}            
            
            Define a hybrid QPT prover $\tilde{\P}$ as follows.
            On input $1^\security$, sample a bit $C\sim\mathrm{Bernoulli}(q_\security)$ using standard rational sampling, and then simulate $\P_C$ in its interaction with the verifier $\V$.
            Because $q_\security$ has only polynomially many bits, this sampling runs in polynomial time, so $\tilde{\P}$ is QPT.
            
            Let $\tilde{P}_\security$ denote the Bell-mapped distribution induced by the interaction of $\tilde{\P}$ with $\V$ at security $\security$.
            Conditioning on the internal coin $C$, we have
            \begin{equation}
              \tilde{P}_\security \;=\; q_\security\,P_0 \;+\; (1-q_\security)\,P_1 \; .
            \end{equation}
            Therefore,
            \begin{equation}
              \delta\!\left(\tilde{P}_\security,\; q P_0 + (1-q)P_1\right)
              \;=\; \frac{1}{2}\,\left\| (q_\security - q)\,(P_0 - P_1) \right\|_1
              \;\leq\; \left| q_\security - q \right| \cdot \delta(P_0,P_1)
              \;\leq\; \left| q_\security - q \right|
              \;\leq\; \negl(\security) \; ,
            \end{equation}
            where we used $\left\|P_0-P_1\right\|_1 = 2\,\delta(P_0,P_1)$ and $\delta(P_0,P_1)\leq 1$.
        \end{proof}

        \begin{remark}[On non-efficient mixing weights]
            If~$q$ is not efficiently computable, the construction may fall outside QPT because producing~$q_{\security}$ to negligible accuracy could take superpolynomial time or may not even be possible at all (e.g., Chaitin’s constant $\Omega$).
        \end{remark}

        Note that, a priori, the computational quantum set~$\quantumcomp{\compLeakage}$ could coincide with the classical computational set~$\localcomp{\compLeakage}$ (e.g., under assumptions that preclude any quantum advantage), in which case the corresponding computational Bell inequalities would be vacuous.
        Nonetheless, under appropriate cryptographic assumptions, $\quantumcomp{\compLeakage}$ is nontrivial: there exist canonical-form protocols and QPT strategies whose Bell-mapped distributions achieve a constant violation of~$\I$, while every PPT prover satisfies $\I(P_\security)\leq \negl(\security)$.
        The following, Subsection~\ref{sec:showcases}, demonstrates this.

    \subsection{Showcases}\label{sec:showcases}        

        In this subsection we illustrate the framework on concrete protocols.
        The goal is to show---at a high level---how to choose a Bell mapping $(\imap,\omap)$, argue the hidden-input property, and then evaluate the induced distribution with a computational Bell inequality.
        Importantly, the same computational Bell inequality~$\I$ applies across the examples in the $\bellscenarioShort=(2,2,2,2)$ scenario; only the leakage parameter $\kappa$ (and the slack $\vartheta$) vary.
        This highlights the modular, plug-and-play nature of the method.

        The following lemma defines the computational Bell inequality~$\I$ for the Bell scenario $\bellscenarioShort=(2,2,2,2)$.
        That is, for any PPT prover~$\P$, if the Bell mapping~$(\imap,\omap)$ satisfies the hidden-input condition with leakage~$\compLeakage$, then the induced Bell-mapped distribution~$P_\security$ achieves at most a negligible violation of~$\I(P_\security)$.
        
        \begin{lemma}\label{lemma:showcase_ineq}
            Let~$\flexibility>0$.
            let~$\local_\flexibility$ be the MDL set for the Bell scenario~$\bellscenarioShort=(2,2,2,2)$, defined in Lemma~\ref{lemma:amdl_mdl_inequality_bound} and let~$\I_\flexibility$ be the corresponding MDL inequality defined in Equation~\eqref{eq:gisin_inequality}:
            \newcommand{\dist}{P_{\MakeUppercase{\oa\ob\ia\ib}}}
               \begin{align}
                    \I_\flexibility(P) \coloneqq
                    &\tfrac{1}{2}\qty(\tfrac{1}{2} - \compLeakage - \flexibility)\dist(0000) \nonumber \\
                    &- \tfrac{1}{2}\qty(\tfrac{1}{2} + \compLeakage + \flexibility)
                    \qty(\dist(0101) + \dist(1010) + \dist(0011)) \; .
                \end{align}
            Then, for any mapped sequence of distributions~$P_\security$ induced by a classical prover~$\P$, the functional~$\I$ defined as
            \begin{equation}\label{eq:comp_ineq_chsh_scenario}
                \I(P)
                \;\coloneqq\;
                \I_\flexibility(P) - \frac{1}{4}\frac{\compLeakage}{\compLeakage+\flexibility}
                \left(
                    \tfrac{1}{2}-\compLeakage-\flexibility
                \right)
                \;,                
             \end{equation}
             satisfies
            \begin{equation}
                \I(P_\security) \leq \negl(\security) \;.
            \end{equation}
            That is, the functional~$\I$, cannot be violated by any classical prover~$\P$ by more than a negligible amount.
        \end{lemma}

        \begin{proof}
            See Appendix~\ref{proof:showcase_ineq}.
        \end{proof}

        The inequality~$\I$ defined above depends only on the Bell-mapped distribution and, in particular, does not depend on the internal structure of any specific canonical protocol or prover.
        To interpret~$\I$ as a computational Bell inequality, that is, one that cannot be violated by any efficient classical prover, we restrict attention to protocols for which the Bell mapping~$(\imap,\omap)$ satisfies the hidden-input condition with leakage~$\compLeakage$ (Definition~\ref{assumption:compLeakage}).
        
        In each case we show that, under the stated cryptographic assumptions, the Bell mapping~$(\imap,\omap)$ satisfies the hidden-input condition with leakage~$\compLeakage$; we then construct the induced Bell distribution~$P$ and show that a quantum prover violates~$\I$ beyond the classical bound, thus certifying quantumness.

        \subsubsection{Protocols based on trapdoor claw-free functions}
        
        In this subsection, we instantiate our framework using the trapdoor claw-free function (TCF) based protocol from~\cite{kahanamoku2022classically}, expressed in canonical form. (Readers who are not familiar with~\cite{kahanamoku2022classically}  should consult Figure~\ref{fig:kcvy_protocol} in Appendix~\ref{ap_sec:kcvy_protocol}).
        We define a Bell mapping for this protocol, verify that it satisfies the hidden input condition (Definition~\ref{assumption:compLeakage}), and thereby show that the computational MDL inequality~$\I$ applies to the induced Bell distribution.
        This sets the stage for analyzing honest quantum strategies that violate the inequality and thus certify quantumness under cryptographic assumptions.

        \begin{definition}\label{def:tcf_bell_mapping}
            Let $\trans = (k,\image, r, d)$ be the transcript in the protocol of~\cite{kahanamoku2022classically}, where~$k$ is the TCF key,~$\image$ is a TCF image,~$r$ and $d$  are a binary strings. (See also Appendix~\ref{ap_sec:kcvy_protocol} for a definition of TCF and Figure~\ref{fig:kcvy_protocol} for an honest implementation of the protocol).
            
            The Bell mapping is defined as
            \begin{equation}\label{eq:tcf_bell_mapping}
                \imap(\trans) = \indicator{r\cdot\preimage_0 = r\cdot \preimage_1} \; ;\quad
                \omap(\trans) =
                \begin{cases}
                    r\cdot\preimage_0 & \text{if } \imap(\trans) = 0 \\
                    d\cdot(\preimage_0 \oplus \preimage_1) & \text{else}
                \end{cases}
                \;,
            \end{equation}
            where~$\qty{\preimage_0,\preimage_1}$ are preimages of~$\image$ with respect to the function~$f_k$.
        \end{definition}

        \begin{lemma}\label{lemma:trapdoor_claw_free_leakage}
            Let~$(\V, \P)$ be verifier-prover pair performing a canonical form protocol based on trapdoor claw-free functions, where~$\P$ is a classical probabilistic polynomial-time (PPT) device.

        Suppose that~$\P$ succeeds in Phase~A of the protocol with probability at least~$1 - 2\compLeakage$.
            Then the virtual input~$\imap(\trans)$, defined via the Bell mapping in Equation~\eqref{eq:tcf_bell_mapping}, is hidden according to Definition~\ref{assumption:compLeakage}.
            That is, 
            \begin{equation}
                \left|
                \underset{\trans}{\E}~\Pr \left( \P(\trans) = \imap(\trans) \right)
                - \frac{1}{|\Ia|}
                \right| \leq \compLeakage + \negl(\security) \;,
            \end{equation}
            where~$\trans$ is the interaction transcript between~$\V$ and~$\P$.
        \end{lemma}

        The proof of the lemma structurally follows the same reduction as in~\cite[Theorem 2]{kahanamoku2022classically}, where the goal was to bound the CHSH score of a classical prover directly.
        Here, however, we apply the reasoning in order to bound the predictability of the virtual input~$\imap(\trans)$.
        As in the analysis of other protocols using our approach, this is the \emph{only} place in which the computational assumption enters the picture.
        This exemplifies the \emph{modularity} of our methods and the \emph{fundamental} understanding it provides by pinning down the relation between the computational assumption and nonsignaling.

        \begin{proof}
            Assume towards contradiction that the hidden input condition does not hold.
            I.e., there exists a non negligible function~$\noneg$ such that
            \begin{equation}
                p_\imap\coloneqq\Pr_{\trans} (\P(\trans) = \imap(\trans)) \geq \frac{1}{|\Ia|} + \compLeakage + \noneg(\security) \;.
            \end{equation}
            We construct an adversary~$\A$, based on~$\P$, that breaks the claw-free property of the trapdoor function used in the protocol.
            
            $\A$ begins by simulating a verifier-prover interaction and challenging~$\P$ for a preimage test to receive a preimage~$\preimage\coloneqq\preimage_0$ of~$\image$ with a success probability of~$p_{\preimage_0}\coloneqq 1-2\compLeakage$.
            $\A$ then rewinds~$\P$, which is possible since~$\P$ is a PPT device, extracts an interaction transcript~$\trans=\preimage,r,d$, and challenges~$\P$ for a guess of the virtual input~$\imap(\trans)$ with a success probability of~$p_\imap$.

            Condition on the event that $\A$ both obtains a valid preimage $\preimage_0$ and correctly predicts the virtual input bit $\imap(\trans)$.
            By the definition of the Bell mapping in Equation~\eqref{eq:tcf_bell_mapping}, $\imap(\trans)$ reveals whether $r\cdot\preimage_0$ equals $r\cdot\preimage_1$.
            Since $\A$ knows $r$ (from $\trans$) and $\preimage_0$, it can compute $r\cdot\preimage_0$ and hence deduce $r\cdot\preimage_1$ via
            \begin{equation}
              r\cdot\preimage_1 =
              \begin{cases}
                r\cdot\preimage_0 & \text{if } \imap(\trans)=1 \;,\\
                1 \oplus (r\cdot\preimage_0) & \text{if } \imap(\trans)=0 \; .
              \end{cases}
            \end{equation}

            Now,~$\A$ proceeds to rewind and challenge~$\P$ for more guesses of the virtual input~$\imap(\trans)$ by querying specific choices of~$r$.
            In particular,~$\A$ is a noisy oracle to the encoding of~$\preimage_1$ under the Hadamard code.
            By Goldreich-Levin~\cite{goldreich1989hardcore}, list decoding applied to such an oracle will generate a polynomial-length list of candidates for~$\preimage_1$.
            If the noise rate of the oracle is noticebly less than~$1/2$,~$\preimage_1$ will be in the list with high probability.
            $\A$ can then iterate through the list and check which candiate satisfies~$f(\preimage_0)=f(\preimage_1)$, thereby breaking the claw-free property of the trapdoor function.

            By~\cite[Lemma 1]{kahanamoku2022classically}, for a particulate iteration of the protcol, the probability that list decoding succeeds is bounded by~$p_{\preimage_1} > 2 p_\imap-1-2\mu$, for a noticeable function~$\mu$ of our choice. 

            \begin{align}
                \Pr(\text{Guessing }\preimage_0 \cap\text{Guessing }\preimage_1) & \geq 1 - (1-p_{\preimage_0}) - (1-p_{\preimage_1}) \\
                & = -2 \compLeakage + 2 p_\imap - 1 - 2\mu \\
                & \geq -2\compLeakage + 1 + 2\compLeakage + 2\noneg - 1 -2\mu \\
                & = 2(\noneg - \mu) \; .
            \end{align}

            By choosing~$\mu = \noneg/2$, we obtain a contradiction, since~$\A$ breaks the claw-free property of the trapdoor function with non-negligible probability. 
        \end{proof}

        The honest implementation (For honest implementation see Figure~\ref{fig:kcvy_protocol} in Appendix~\ref{ap_sec:kcvy_protocol}) of the TCF based protocol violates the inequality~\eqref{eq:comp_ineq_chsh_scenario} for certain parameter choices.
        Specifically, for~$\compLeakage = 0.025$ and~$\flexibility = \compLeakage^{0.45}$, the induced Bell distribution~$P$ satisfies a constant
        \begin{equation}
            \I(P) \approx  3.7 \cdot 10^{-4} > \negl(\security) \;,
        \end{equation}
        for any security parameter~$\security$, thereby certifying quantum behavior under the computational MDL framework.

        \subsubsection{Protocols based on compiled games}

        Compiled nonlocal games~\cite{kalai2022compiled} transform a standard Bell test (e.g., CHSH) into a single-prover cryptographic protocol.
        The verifier samples inputs~$\ia$ and~$\ib$ for the two players and sends the prover an encryption~$\tilde{\ia}\!\coloneqq\!\Enc(\ia)$ of~$\ia$ under a quantum homomorphic encryption (QHE) scheme\footnote{For specific works relating QHE, see~\cite{gupte2024qfhe,brakerski2018qfhe}.} (see Appendix~\ref{ap_sec:compiled} for the definition).
        Using the homomorphic encryption scheme, the prover computes an encryption~$\widetilde{\oa}$ of the answer~$\oa$ that the respective player would output on input~$\ia$.
        The prover also receives~$\ib$ unencrypted and computes the corresponding answer~$\ob$ for the respective player.
        The important thing about the homomorphic encryption scheme, is that it allows the prover to simulate both parties in the nonlocal game \emph{without} knowing either~$\ia$ nor~$\oa$.

        In this setting, the Bell mapping is natural: define
        \begin{equation}
            \imap(\trans) \coloneqq \Dec(\widetilde{\ia})
            \qquad\text{and}\qquad
            \omap(\trans) \coloneqq \Dec(\widetilde{\oa}),
        \end{equation}
        where $\widetilde{\oa}$ is the encryption produced by homomorphic evaluation of the first player's response circuit on $\Enc(\ia)$.
        By the definition of the QHE, $\Dec(\widetilde\ia)=\ia$ and $\Dec(\widetilde{\oa})=\oa$.
        While efficient decryption uses the secret key held by the verifier, for the purposes of the Bell mapping it suffices that these are well-defined functions of the transcript; they can be viewed as applied by the referee who possesses the key, or simply as mathematically defined (possibly inefficient) maps guaranteed by correctness.
        This choice aligns exactly with the semantics of the compiled game: the virtual input is the first player's question and the virtual output is that player's answer.

        An important feature of this canonical-form protocol is that it allows the immediate translation of any quantum strategy for the original nonlocal game with a respective distribution in the set~$\quantum$, into a valid strategy of a canonical form protocol with a respective distribution in the computational quantum set~$\quantumcomp{\compLeakage}$.
        In particular, quantum violations of Bell inequalities (such as CHSH) are preserved in this setting.
        We refer the reader to~\cite{kalai2022compiled} for a full description of the compilation framework, and to~\cite{brakerski2023simple} for a detailed analysis of CHSH in this context.

        In what follows, we go beyond this prior work by allowing the homomorphic encryption scheme to leak a small amount of information, i.e.\ $\compLeakage>0$.
        This lets us study protocols whose computational soundness is slightly degraded while still exhibiting strong quantum violations.
        Because the compilation preserves any quantum strategy, we can use the  quantum strategy for standard (non-compiled) nonlocal games, from~\cite[Eq.~(6)]{putz2014}, into the compiled setting.
        That strategy violates the MDL inequality stated in Claim~\ref{claim:putz2014}, and consequently gives a compiled quantum strategy that violates the computational inequality in Equation~\eqref{eq:comp_ineq_chsh_scenario}.

        Using the strategy from~\cite[Eq.~(6)]{putz2014}, for example, with~$\compLeakage=0.02$ and~$\flexibility=\compLeakage^{0.19}$, for every security parameter~$\security$, we obtain a constant violation of
        \begin{equation}
            \I(P) \approx 1.45 \cdot 10^{-6} > \negl(\security) \;.
        \end{equation}

\section{Computational-SoC hierarchy}\label{sec:comp_soc_hier}

    In Section~\ref{sec:com_soc}, we showed how a verifier-prover interaction in the canonical protocol can be reformulated in terms of a computational SoC.
    This abstraction provides a cleaner and more conceptual view of tests of quantumness based on computational assumptions.
    As illustrated by the showcases in Section~\ref{sec:showcases}, this approach allows for sharper distinctions between classical (PPT) and quantum (QPT) provers.

    We now build on this perspective to develop tools inspired by nonlocal games, adapted to the computational setting.
    In particular, we define a hierarchy of semidefinite relaxations, a computational analogue of the NPA hierarchy~\cite{navascues2008npa,pironio2010npa2}, designed to approximate the set of correlations achievable by efficient quantum provers \emph{from the outside}.
    Each level of the hierarchy defines a computationally sound outer relaxation of~$\quantumcomp{\compLeakage}$, capturing all feasible QPT strategies while potentially including additional points.
    This hierarchy will be used to prove a computational version of Tsirelson's bound and to derive entropy bounds for canonical protocols.

    There are existing works~\cite{cui2025seqnpa,klep2025seqnpa} that use hierarchy-based relaxations in a similar spirit, within the setting of compiled non-local games.
    Our framework applies to a broader class of canonical form protocols and, in addition, differs in two key ways: (i)~the mechanism by which nonsignaling is enforced, and (ii)~we explicitly relax nonsignaling by allowing bounded virtual-input leakage (i.e., limited signaling).
    
    \subsection{Computational nonsignaling}

    A central difference between the standard device-independent and computational settings lies in how nonsignaling is enforced.
    In traditional Bell scenarios, nonsignaling follows from the bipartite structure of the system: each party's measurements act on separate subsystems, so their local outputs cannot depend on the other party's input.
    In the NPA hierarchy, this is captured either through tensor products or by imposing commutation between measurements on different parties.
    In our setting, where the prover is a single device, nonsignaling cannot be enforced information-theoretically---the prover's internal state typically contains enough information to recover the virtual input.
    Instead, as we demonstrated, signaling is prevented computationally: QPT provers cannot efficiently determine~$\imap(\trans)$.
    The following lemma formalizes this approximate nonsignaling behavior.

    \begin{lemma}[Computational  Nonsignaling]\label{lemma:computational_nonsignaling}
        Assume that the virtual input~$\ia = \imap(\trans)$ is distributed uniformly over~$\Ia$.       
        Then, there exists a computational signaling parameter~$\compSignaling = \O(\compLeakage)$ such that for all~$\ia, \ia' \in \Ia$ and all binary-output QPT algorithms~$\A$ with advice, we have
        \begin{equation}
            \left|
                \Pr\left( \A(\state^{\trans}) = 1 \mid \imap(\trans) = \ia \right) -
                \Pr\left( \A(\state^{\trans}) = 1 \mid \imap(\trans) = \ia' \right)
            \right| \leq \compSignaling + \negl(\security) \; .
        \end{equation}
    \end{lemma}
    
    \begin{proof}
        Assume toward contradiction that the claim does not hold for a signaling parameter~$\compSignaling\coloneqq|\Ia|\cdot\compLeakage$.
        Then there exist $\ia, \ia' \in \Ia$, a QPT algorithm~$\A$ with advice, and a non-negligible function~$\noneg(\security)$ such that
        \begin{equation}\label{eq:signaling_assumption}
            \left|
            \Pr\left( \A(\state^{\trans}) = 1 \mid \imap(\trans) = \ia \right) -
            \Pr\left( \A(\state^{\trans}) = 1 \mid \imap(\trans) = \ia' \right)
            \right| \geq \compSignaling + \noneg(\security) \; ,
        \end{equation}
        for infinitely many values of~$\security$.
    
        We construct a QPT algorithm~$\A'$ that attempts to guess $\imap(\trans)$.
        On input~$\state^{\trans}$, the algorithm~$\A'$ runs~$\A(\state^{\trans})$ and:
        \begin{itemize}[leftmargin=1.5em]
            \item outputs~$\ia$ if $\A(\state^{\trans}) = 1$,
            \item outputs~$\ia'$ otherwise.
        \end{itemize}
    
        The success probability of~$\A'$ is given by
        
        \begin{align}
            \Pr\left( \A'(\state^{\trans}) = \imap(\trans) \right)
            &= \Pr\left( \A'(\state^{\trans}) = \imap(\trans) \mid \imap(\trans) \in \{\ia, \ia'\} \right) \Pr\left( \imap(\trans) \in \{\ia, \ia'\} \right) \nonumber \\
            &\quad + \Pr\left( \A'(\state^{\trans}) = \imap(\trans) \mid \imap(\trans) \notin \{\ia, \ia'\} \right) \Pr\left( \imap(\trans) \notin \{\ia, \ia'\} \right) \; .
        \end{align}

        Conditioned on $\imap(\trans) \in \{\ia, \ia'\}$, the success probability of~$\A'$ is
        \begin{align}\label{eq:conditioned_guessing_probability}
            \Pr(\A'(\state^{\trans})=\imap(\trans)\mid\imap(\trans)\in\qty{\ia,\ia'})
            & = \Pr(\imap(\trans)=\ia\mid\imap(\trans)\in\qty{\ia,\ia'})
            \Pr(\A(\state^{\trans})=\imap(\trans)\mid\imap(\trans)=\ia) \nonumber \\
            &\quad +
            \Pr(\imap(\trans)=\ia'\mid\imap(\trans)\in\qty{\ia,\ia'})
            \Pr(\A(\state^{\trans})=\imap(\trans)\mid\imap(\trans)=\ia')\;.
        \end{align}
        Assume WLOG that the virtual inputs are symmetrized through public randomness sampled by the verifier and that is part of the transcript.
        I.e.,~$\Pr(\imap(\trans)=\ia)={1}/{|\Ia|}$.
        Then, Equation~\eqref{eq:conditioned_guessing_probability} becomes
        \begin{equation}
            \Pr(\A'(\state^{\trans})=\imap(\trans)\mid\imap(\trans)\in\qty{\ia,\ia'})
            \!=\!
            \frac{1}{2} + \frac{1}{2}\!
            \left(
            \Pr\!\left( \A'(\state^{\trans}) = 1 \mid \imap(\trans) = \ia \right) -
            \Pr\!\left( \A'(\state^{\trans}) = 1 \mid \imap(\trans) = \ia' \right)
            \right) \; .
        \end{equation}
    
        Since $\imap(\trans)$ is uniformly distributed over~$\Ia$, we have
        \begin{equation}
            \Pr\left( \imap(\trans) \in \{\ia, \ia'\} \right) = \frac{2}{|\Ia|} \; .
        \end{equation}
    
        Thus, the advantage of~$\A'$ over random guessing is at least
        \begin{equation}
            \frac{1}{|\Ia|} \cdot \left(
            \Pr\left( \A(\state^{\trans}) = 1 \mid \imap(\trans) = \ia \right) -
            \Pr\left( \A(\state^{\trans}) = 1 \mid \imap(\trans) = \ia' \right)
            \right) \; .
        \end{equation}
    
        By Equation~\eqref{eq:signaling_assumption}, this is at least
        \begin{equation}
            \frac{1}{|\Ia|} \cdot \left( \compSignaling + \noneg(\security) \right) \; .
        \end{equation}
    
        Setting~$\compSignaling \coloneqq |\Ia| \cdot \compLeakage$ yields an advantage of at least $\compLeakage + \noneg(\security)/|\Ia|$, which contradicts Definition~\ref{assumption:compLeakage}.
    
        Since $\noneg(\security)$ is non-negligible and $|\Ia|$ is constant, the term $\noneg(\security)/|\Ia|$ remains non-negligible.
        Therefore, the lemma holds.
    \end{proof}
    
    This lemma justifies the approximate nonsignaling constraint in our hierarchy below, formally defined in the next subsection.
    It ensures that differences in acceptance probability under distinct inputs~$\ia$ and~$\ia'$ are bounded by~$\compSignaling$, which is itself controlled by the leakage~$\compLeakage$.

    \subsection{Defining the hierarchy}
        \newcommand{\bdset}{\Pi}
        \newcommand{\socH}{\mathsf{CSoC}}

        We begin by identifying the types of measurement sequences an efficient prover can apply to their internal state.
        In our setting, a prover may perform a sequence of measurements, one after another, where each measurement may depend on the outcome of the previous ones.
        These outcome-dependent strategies can be viewed as branching programs over a tree of sequential measurement steps.
        To model this structure, we define a restricted set of test operators that simulate such adaptive behavior.




        \begin{definition}[Valid adaptive measurement programs]\label{def:adaptive_measurement_paths}
            Fix a family of POVMs $\povms$ acting on a Hilbert space $\mathcal{H}$.
            For $\ell\in\N$, define the set $\bdset_\ell$ of test operators realizable by adaptive programs of depth at most $\ell$ inductively:
            \begin{itemize}
                \item \textbf{Base case:}
                    \begin{equation}
                        \bdset_0 \coloneqq \qty{ \id } \; .
                    \end{equation}
                \item \textbf{Inductive step:} given $\bdset_\ell$, set
                    \begin{equation}
                        \bdset_{\ell+1} \coloneqq
                        \left\{
                        \sum_{\ob\in S} \polynomialB^{(\ob)}\,\measby \;:\;
                        \ib\in\Ib,\; S\subseteq\Ob,\; \polynomialB^{(\ob)}\in\bdset_\ell \text{ for all } \ob\in S
                        \right\} \; .
                    \end{equation}
            \end{itemize}
        \end{definition}
        
        \noindent An operator $\polynomialB\in\bdset_\ell$ encodes a branching quantum measurement program of depth at most $\ell$: first measure the POVM for some input~$\ib$, keep only outcomes in~$S$, and---conditional on outcome~$\ob\in S$---continue with the depth-$\ell$ subprogram represented by~$\polynomialB^{(\ob)}$.
        By closing under this inductive rule,~$\bdset_\ell$ is exactly the family of Kraus operators obtained by composing the given POVMs with classical postselection and branching for at most $\ell$ rounds.
        
        If the device can implement each POVM~$\{\measby\}_{\ob}$ and perform classical control on the observed outcome, then any~$\polynomialB\in\bdset_\ell$ can be realized efficiently as an~$\ell$-step adaptive test: in round~1 measure the chosen~$\ib$, abort if the outcome is not in~$S$, otherwise record~$\ob$ and proceed with the round-2 subprogram prescribed by~$\polynomialB^{(\ob)}$; continue for at most~$\ell$ rounds and finally output~1 (accept).
        For any state~$\state$, the acceptance probability of this procedure equals
        \begin{equation}
          \tr\!\left[ \polynomialB \,\state\, \polynomialB^\dagger \right] \; ,
        \end{equation}
        so it is a well defined probability in~$[0,1]$.
        This implementation uses at most~$\ell$ POVM applications and classical branching, with running time polynomial in~$\ell$ and in the circuit sizes of the POVMs, and it does not rely on indirect or nonphysical operations.
        
        \begin{definition}[Level-$\ell$ computational-SoC hierarchy ($\socH_\ell{(\compSignaling)}$)]\label{def:computational_soc}
            Fix a level~$\ell \in \mathbb{N}$. The level-$\ell$ computational SoC relaxation, denoted~$\socH_\ell$, is defined as the set of distributions~$P$ over~$(\oa, \ob, \ia, \ib) \in \bellscenarioDistSet$ for which there exist:
            \begin{itemize}
                \item subnormalized quantum states~$\{ \state^{\oa \mid \ia} \}_{\ia \in \Ia,\, \oa \in \Oa}$,
                \item POVMs~$\{{\{\measby\}}_{\ob\in\Ob}\}_{\ib\in\Ib}$,
            \end{itemize}
            satisfying:
            \begin{enumerate}[label=(\roman*)]
                \item \textbf{Completeness:} for all~$\ib \in \Ib$,
                \begin{equation}\label{eq:csoc_i}
                    \sum_{\ob \in \Ob} \meas{\ib}{\ob} = \id \; .
                \end{equation}

                \item \textbf{Positivity:} for all~$\ia, \oa, \ib, \ob$,
                \begin{equation}\label{eq:csoc_ii}
                    \state^{\oa \mid \ia} \succeq 0
                    \quad \text{and} \quad
                    \meas{\ib}{\ob} \succeq 0 \; .
                \end{equation}

                \item \textbf{Reproduction of observed correlations:}
                \begin{equation}\label{eq:csoc_iii}
                    P(\oa, \ob, \ia, \ib) =
                    \frac{1}{|\Ia||\Ib|} \cdot
                    \tr\left( \meas{\ib}{\ob} \, \state^{\oa \mid \ia} \right) \; .
                \end{equation}

                \item \textbf{Approximate nonsignaling under adaptive strategies:} 
                For all operators~$\polynomialB \in \bdset_\ell$ and all~$\ia, \ia' \in \Ia$, we require:
                \begin{equation}\label{eq:csoc_iv}
                \left|
                    \tr( \polynomialB \state^{\ia} \polynomialB^\dagger )
                    -
                    \tr( \polynomialB \state^{\ia'} \polynomialB^\dagger )
                \right|
                \leq \compSignaling \; .
                \end{equation}
            \end{enumerate}
        \end{definition}

        We remark that our nonsignaling condition in Equation~\eqref{eq:csoc_iv}, formulated in terms of test operators from~$\bdset_\ell$, has a key operational advantage over prior work such as~\cite{klep2025seqnpa,kulpe2024boundquantumvaluecompiled}.
        Firstly, in prior works, the nonsignaling condition is enforced exactly (i.e., with~$\compLeakage = 0$).
        Secondly, and more importantly, the nonsignaling condition is stated in terms of the expectations of general noncommutative monomials in the prover's measurement operators,
        which need not correspond to physically realizable operations.
        To justify that a QPT prover can simulate such expectations,~\cite{klep2025seqnpa,kulpe2024boundquantumvaluecompiled} construct a block-encoding argument to show that certain expected values are accessible to the prover via indirect measurements.
        This adds a layer of technical overhead to the soundness proof.
        In contrast, our test operators~$\polynomialB \in \bdset_\ell$ correspond directly to physically realizable measurement programs.
        As a result, our soundness condition is justified by construction, and does not require any indirect access argument.

        The following lemma shows that the hierarchy~$\socH_\ell$ provides an outer approximation to the set of quantum correlations achievable by an efficient prover.
        We refer to this property as the ``soundness'' of the hierarchy.
        This is appropriate in our setting, since one is typically interested in worst-case guarantees, and outer approximations allow us to upper-bound the behavior of all efficient quantum strategies.
        Moreover, if one additionally imposes that the measurement operators commute, then the same structure also yields an outer approximation to the set of classical strategies.
        We refer the reader to Sections~\ref{sec:cTs} and~\ref{sec:ent_cert} for explicit examples where this property is used.
        
        \begin{lemma}[Soundness of the computational SoC hierarchy]\label{lemma:cnpa_soundness}
            Let~$(\mathcal{V}, \mathcal{P})$ be a verifier-prover pair performing a canonical form protocol (Definition~\ref{def:canonical_form}), and let~$\bellmap$ be a Bell mapping.
            Let~$P$ be the distribution over tuples~$(\oa,\ob,\ia,\ib) \in \bellscenarioDistSet$ induced by the interaction of~$(\mathcal{V}, \mathcal{P})$ and the mapping~$\bellmap$.

            Then for any level~$\ell \in \N$, the distribution~$P$ belongs to the level-$\ell$ computational NPA relaxation~$\cnpal$, up to additive error~$\negl(\security)$ and signaling parameter~$\compSignaling = \O(\compLeakage)$.
        \end{lemma}

        \begin{proof}
            Fix a canonical-form interaction $(\V,\P)$ and Bell mapping~$\bellmap$, and let $P$ be the induced Bell-mapped distribution.
            Use as witnesses the states $\{\state^{\oa\mid\ia}\}_{\ia,\oa}$ from Definition~\ref{def:virtual_state_family} and the Phase~B POVMs $\{\measby\}_{\ib,\ob}$.
            
            (i) Completeness and (ii) Positivity are immediate: $\sum_{\ob}\meas{\ib}{\ob}=\mathbbm{I}$ for each $\ib$, every $\meas{\ib}{\ob}\succeq 0$, and each $\state^{\oa\mid\ia}$ is a convex combination of positive states, hence $\state^{\oa\mid\ia}\succeq 0$.
            
            (iii) Reproduction follows from Lemma~\ref{lem:state_representation}:
            \begin{equation}
            P(\ia,\ib,\oa,\ob)\;=\;\frac{1}{|\Ia||\Ib|}\,\Tr\!\left[\measby\,\state^{\oa\mid\ia}\right]\,.
            \end{equation}
            
            (iv) For any $\polynomialB\in\bdset_\ell$, Definition~\ref{def:adaptive_measurement_paths} ensures $\polynomialB$ is a physically realizable depth-$\ell$ adaptive test built from $\{\measby\}$.
            Hence, for all $\ia,\ia'\in\Ia$,
            \begin{equation}
            \left|
            \Tr\!\left[\polynomialB\,\state^{\ia}\,\polynomialB^\dagger\right]
            -
            \Tr\!\left[\polynomialB\,\state^{\ia'}\,\polynomialB^\dagger\right]
            \right|
            \leq \compSignaling+\negl(\security)
            \end{equation}
            by Lemma~\ref{lemma:virtual_input_indistinguishability}, with $\compSignaling=\O(\compLeakage)$.
            
            Therefore, every~$P \in \quantumcomp{\compLeakage}$ lies in the closure of~${\socH_\ell(\compSignaling)}$ with~$\compSignaling=\O(\compLeakage)$.
        \end{proof}
        
        \begin{lemma}[Approximate nonsignaling of adaptive strategies]\label{lemma:virtual_input_indistinguishability}
            Let~$(\V, \P)$ be a verifier-prover pair performing a canonical form protocol, and let~$\bellmap$ be a Bell mapping.
            Assume the hidden input condition (Definition~\ref{assumption:compLeakage}) holds with leakage~$\compLeakage$.
            Let~$\povms$ denote the prover's measurement operators, and let~$\polynomialB \in \bdset_\ell$ be a valid test operator of depth at most~$\ell$.

            Then for all~$\ia, \ia' \in \Ia$, we have:
            \begin{equation}
                \left|
                    \tr( \polynomialB \state^{\ia} \polynomialB^\dagger ) -
                    \tr( \polynomialB \state^{\ia'} \polynomialB^\dagger )
                \right| \leq \compSignaling + \negl(\security) \; ,
            \end{equation}
            for some signaling parameter~$\compSignaling = \O(\compLeakage)$.
        \end{lemma}

        \begin{proof}
            This follows directly from Lemma~\ref{lemma:computational_nonsignaling}: 
            a distinguisher between $\state^\ia$ and $\state^{\ia'}$ via any adaptive operator~$\polynomialB$ 
            would yield a QPT distinguisher between hidden inputs, contradicting the hidden-input assumption. 
        \end{proof} 

        One may wonder whether the hierarchy converges to the exact set as $\ell\rightarrow\infty$.
        In practice (as will be shown by the examples bellow), the converges and its rate do not matter as we usually get strong, even tight, results from low levels of the hierarchy.
        Nevertheless, this question is of mathematical nature and was studied in the context of nonlocal games~\cite{navascues2008convergent,pironio2010npa2}.
        In the context of protocols with computational assumptions, similar questions were addressed for protocols based on compiled games~\cite[Theorem~6.1]{kulpe2024boundquantumvaluecompiled} for the case~$\compLeakage=0$, which generalizes to our case.
        In what follows, we do not need to assume anything regarding the convergence of the hierarchy.
        We get provably tight results already in the second level of the hierarchy.

    \subsection{Application: computational Tsirelson's bound}\label{sec:cTs}

        In this subsection, we illustrate how the computational-SoC hierarchy can be used to upper-bound the CHSH score achievable by a QPT prover in a canonical protocol, under computational leakage~$\compLeakage$.
        This yields a computational analogue of Tsirelson's bound: a leakage-dependent upper limit on the quantum value, grounded in computational hardness rather than physical commutation.

        We formalize this idea as a semidefinite program (SDP) that optimizes the CHSH value over the level-$\ell$ relaxation~$\socH_\ell$ of the computational-SoC hierarchy.
        The SDP enforces a constraint on computational signaling: expectations of valid test operators (see Definition~\ref{def:adaptive_measurement_paths}) must be approximately independent of the virtual input~$\ia$.
        The allowed deviation is governed by the signaling parameter~$\compSignaling = \O(\compLeakage)$, as discussed in Section~\ref{sec:comp_soc_hier}.
        We find that the bounds produced by this SDP capture the optimal classical and quantum CHSH values under computational leakage.
        This is illustrated in Figure~\ref{fig:kcvy_tsirelson}, which compares our results against known analytical bounds.
        All SDPs were modeled and solved using the \texttt{ncpol2sdpa} library~\cite{ncpol2sdpa}.

        Figure~\ref{prog:comp_chsh} defines the SDP used to compute the computational Tsirelson bound.
        The program maximizes the CHSH winning probability over a collection of state-measurement pairs~$\{ \state^{\oa \mid \ia} \}$ and~$\{ \meas{\ib}{\ob} \}$, subject to standard positivity and normalization constraints, as well as approximate nonsignaling.
        The constraint labeled ``Approximate signaling'' enforces computational input-independence: for each valid test operator~$\polynomialB$ of degree at most~$\ell$, the expectation~$\tr(\polynomialB^\dagger \polynomialB \state^{\ia})$ must not vary significantly between different inputs~$\ia, \ia' \in \Ia$.
        This models the fact that~$\ia$ is hidden under a computational assumption, as discussed in Section~\ref{sec:comp_soc_hier}.
        The CHSH objective is written in correlator form using~$(-1)^{\ia \cdot \ib + \oa + \ob}$, consistent with the canonical game.

        \begin{figure}
            \centering
            \begin{program}
                \textbf{Program: Computational CHSH SDP (level-$\ell$)}
                
                \vspace{-1em}
                
                \begin{align*}
                \text{Variables:} \quad
                & \{ \state^{\oa \mid \ia} \}_{\ia, \oa}, \quad \{ \meas{\ib}{\ob} \}_{\ib, \ob} \\[0.5em]
                \text{Objective:} \quad
                & \text{maximize } \sum_{\ia, \ib, \oa, \ob}
                {(-1)}^{\ia \cdot \ib + \oa + \ob}
                \cdot \tr\left( \meas{\ib}{\ob} \state^{\oa \mid \ia} \right) \\[0.5em]
                \text{Subject to:} \quad
                & \state^{\oa \mid \ia} \succeq 0 \\
                & \tr\left( \sum_{\oa} \state^{\oa \mid \ia} \right) = 1 \quad \text{for all } \ia \\
                & \meas{\ib}{\ob} \succeq 0 \\
                & \sum_{\ob} \meas{\ib}{\ob} = \id \quad \text{for all } \ib \\[0.5em]
                \text{Approximate signaling:} \quad
                & \left|
                \tr\Big( \polynomialB^\dagger \polynomialB \state^{\ia} \Big) -
                \tr\Big( \polynomialB^\dagger \polynomialB \state^{\ia'} \Big)
                \right| \leq \compSignaling \\[-0.5em]
                & \text{~~for all } \ia, \ia' \in \Ia, \; \text{and al } \polynomialB \in \bdset_\ell
                \end{align*}
            \end{program}
            \caption{Level-$\ell$ computational CHSH semidefinite program.}
            \label{prog:comp_chsh}
        \end{figure}

        \begin{figure}[t]
            \centering
            \includegraphics[scale=0.5]{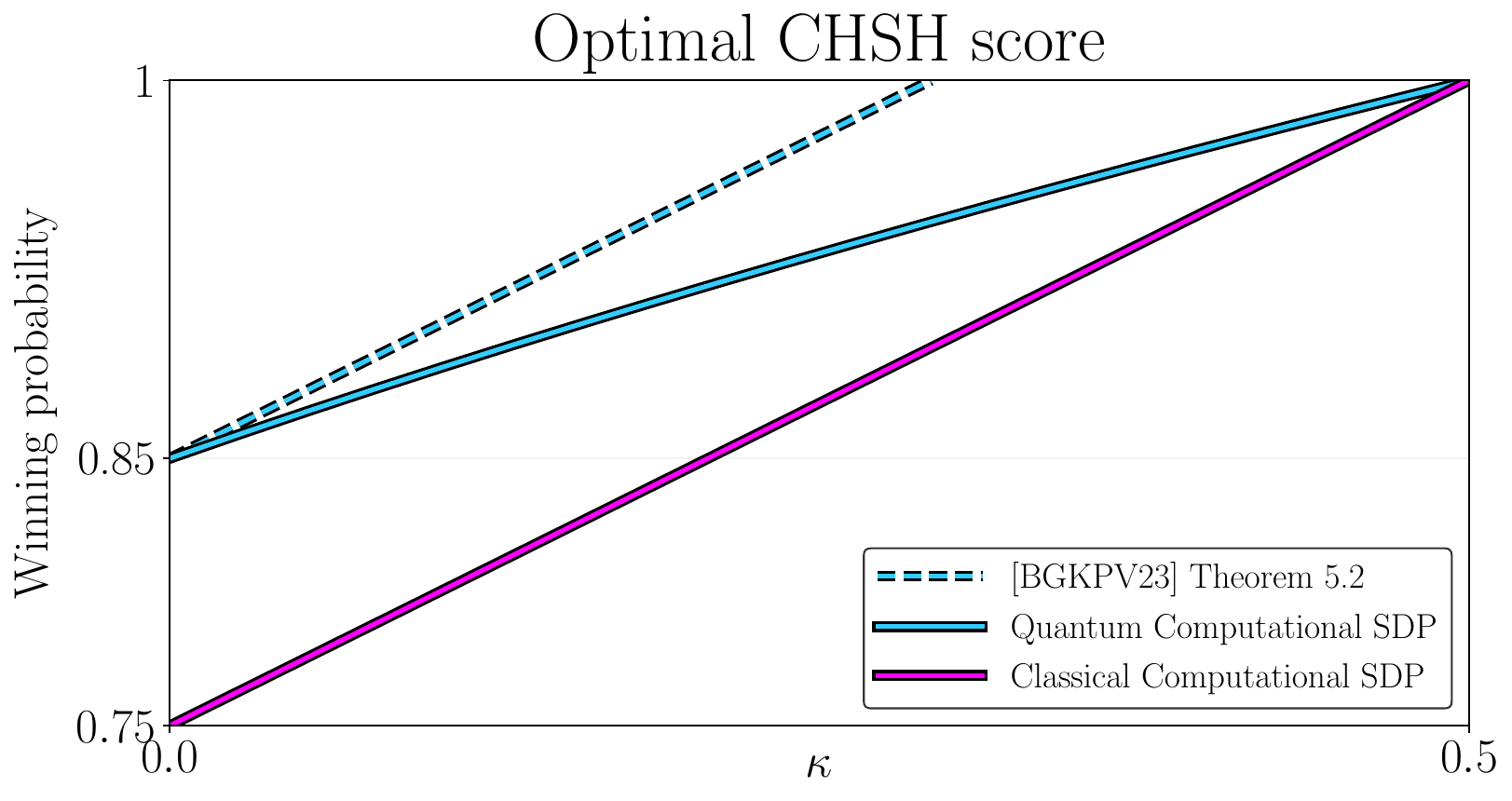}
            \caption[Computational Tsirelson bounds vs leakage.]
            {
                \footnotesize
                Computational Tsirelson bounds for the CHSH game as a function of the leakage parameter~$\compLeakage$.
                The cyan curve shows the optimal winning probability over the computational quantum set, computed using the level-2 relaxation~$\socH_{\ell=2}$.
                The dashed cyan line is the analytic upper bound from~\cite[Theorem~5.2]{brakerski2023simple}, derived specifically for the KCVY protocol.
                The magenta curve shows the optimal winning probability over the computational classical set, also computed via~$\socH_{\ell=2}$.
                It exactly matches the classical bound from~\cite[Theorem~2]{kahanamoku2022classically}.
                As~$\compLeakage \to 0$, the quantum and classical curves converge to the standard Tsirelson and classical CHSH values, $\cos^2(\pi/8) \approx 0.85$ and~$0.75$ respectively.
                As~$\compLeakage \to 0.5$, both bounds converge to~$1$.
            }
            \label{fig:kcvy_tsirelson}
        \end{figure}

        \paragraph{Optimal classical bound.}
            The magenta curve in Figure~\ref{fig:kcvy_tsirelson} represents the maximum CHSH score achievable by any classical prover in a canonical protocol, given computational leakage~$\compLeakage$.
            This bound is computed via the level-2 relaxation of the hierarchy~$\socH_2$, with the additional constraint that the measurement operators commute.
            The result exactly matches the analytic classical bound derived in~\cite[Theorem~2]{kahanamoku2022classically}, confirming that the SDP captures the classical behavior precisely.

            As expected, the winning probability increases monotonically with~$\compLeakage$:
            At~$\compLeakage = 0$, one recovers the standard local value of the CHSH game,~$0.75$.
            As~$\compLeakage \to 0.5$, classical provers can fully reconstruct the virtual input and simulate any strategy, approaching the maximum winning probability of~$1$.

        \paragraph{Optimal quantum bound.}
            The cyan curve shows the CHSH score achievable by a QPT prover in a canonical protocol, subject to leakage~$\compLeakage$.
            This is again computed via~$\socH_2$, with signaling bounded by~$\compSignaling = 2\compLeakage$ and no restriction to commutative measurements.
            The resulting values outperform the analytic bound from~\cite[Theorem~5.2]{brakerski2023simple}, derived for the specific KCVY protocol, despite our SDP being protocol-agnostic.\footnote
            {
                The original statement in~\cite[Theorem~5.2]{brakerski2023simple} claims a violation of the CHSH inequality by quantum provers that scales as~$\O(\compLeakage^{0.5})$.
                However, as noted in private correspondence with the authors, the method in that proof actually yields a linear bound~$\O(\compLeakage)$, which is stronger.
                We plot the correct (linear) rate here.
            }

            This suggests that in the $(2,2,2,2)$ Bell scenario, protocol-specific structure does not appear to yield stronger bounds.

            It suffices to identify a virtual input~$X$ hidden under leakage~$\compLeakage$, and the hierarchy captures the optimal quantum behavior.
            The SDP limit matches the standard Tsirelson bound~$\cos^2(\pi/8) \approx 0.8535$ as~$\compLeakage \to 0$.
            As leakage increases, the quantum bound also increases and asymptotically reaches~$1$.

        \paragraph{Why level 2 is necessary and sufficient.}
            The need for level-2-style constraints is motivated by an attack described in~\cite{natarajan2023bounding}, where a quantum prover performs two specific measurements in sequence and uses the results to recover the virtual input~$\ia$.
            In particular, they show that if the prover can measure both~$\provermeasurement_{0}$ and~$\provermeasurement_{1}$ adaptively, it can learn~$\ia$ with non-negligible probability, violating the computational hiding assumption.\footnote{\cite{natarajan2023bounding} discusses the compiled nonlocal game with the Bell scenario of~$\bellscenarioShort=(2,2,2,2)$ and~$\compLeakage=0$.}
            While their analysis does not use the language of our hierarchy, it shows that allowing even limited sequential access to measurements can compromise soundness.
            This provides evidence that level~2 of our computational-SoC hierarchy enforces a sufficiently strong constraint: test operators in~$\bdset_2$ model precisely this kind of adaptive access, and allow us to rule out such attacks via approximate nonsignaling constraints.

            To illustrate this concretely, we consider the following strategy:
            define the virtual-input-conditioned states~$\state^{\oa \mid \ia} = \frac{1}{2}\dyad{\oa, \ia}$ and measurements
            \begin{equation}
                \measby = \sum_{\oa, \ia} \indicator{\oa\oplus\ob = \ia \cdot \ib} \dyad{\oa, \ia} \;,
            \end{equation}
            where~$\indicator{\cdot}$ denotes the indicator function.
            This strategy achieves a perfect CHSH score, and the measurement~$\measby$ alone leaks no information about~$\ia$.
            However, if the measurement outcome~$\oa$ is revealed after the interaction, then~$\ia$ becomes fully determined.
            This form of leakage, undetectable at level~1, is naturally ruled out at level~2, due to the structure of sequential test operators~$\polynomialB \in \bdset_2$.

    \subsection{Application: entropy certification}\label{sec:ent_cert}
       
        We now demonstrate how the computational-SoC hierarchy can be used to certify entropy generated by a quantum prover in a canonical form protocol.
        Figure~\ref{fig:entropy} shows a lower bound on the conditional min-entropy of the prover's output~$B$, given fixed verifier input~$Y = 0$, fixed Bell-mapped values~$X = 0$, $A = 0$, and adversarial side information~$E$.
        The bounds were computed by optimizing over the level-$\ell=2$ relaxation~$\socH_\ell$ with~$\compLeakage=0$ constraints on leakage and signaling.
        Despite the protocol-specific nature of the Bell mapping, the resulting entropy curve matches the standard CHSH entropy bounds -- confirming that our computational framework faithfully captures quantum unpredictability.

        \begin{figure}[ht]
            \centering
            \includegraphics[scale=0.5]{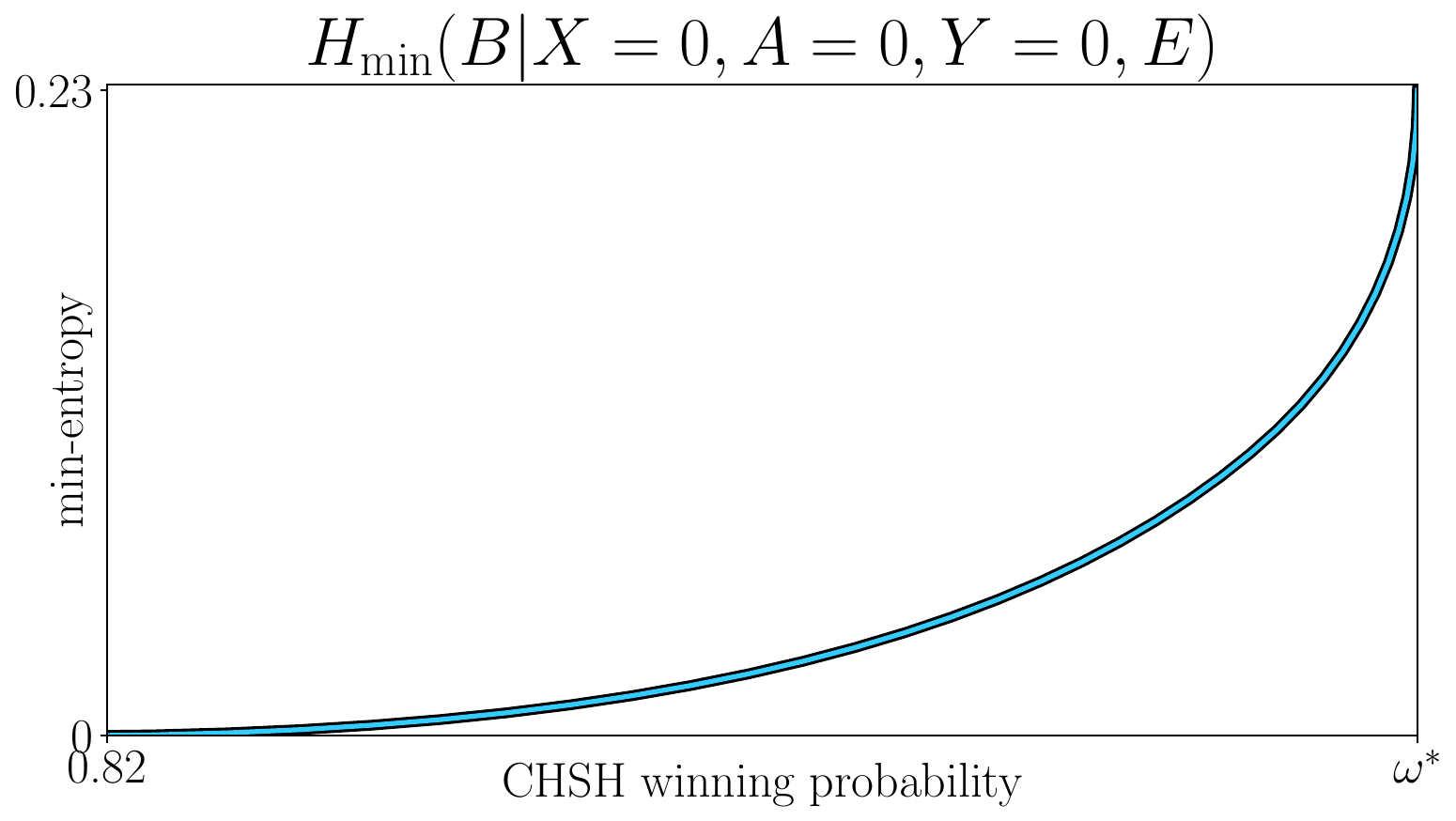}
            \caption
            {
                \footnotesize
                Conditional min-entropy $H_{\min}(B \mid X = 0, A = 0, Y = 0, E)$ as a function of the CHSH winning probability.
                The entropy quantifies the unpredictability of the prover's response~$B$ given fixed Bell-mapped values~$X = 0$, $A = 0$, verifier input~$Y = 0$, and adversarial side information~$E$.
                The curve was computed using the computational-SoC hierarchy under appropriate constraints on leakage and signaling, and matches known entropy values from the standard CHSH scenario.
            }
            \label{fig:entropy} 
        \end{figure}

        It is worth emphasizing why this specific min-entropy quantity is relevant in our setting.
        In standard (non-computational) Bell scenarios, one typically analyzes~$H_{\min}(B \mid X = 0, Y = 0)$, where~$B$ is one party's output conditioned on fixed inputs.
        However, in the canonical protocol setting, if the transcript~$\trans$ is assumed to be public, then both~$X = \imap(\trans)$ and~$A = \omap(\trans)$ become computable from~$\trans$, at least for an inefficient adversary.
        As a result, a computational adversary could in principle obtain this side information, even if not efficiently.
        To account for this, we condition not only on~$Y = 0$ but also on the Bell-mapped values~$X = 0$ and~$A = 0$.
        This leads us to analyze~$H_{\min}(B \mid X = 0, A = 0, Y = 0, E)$, which quantifies the unpredictability of the prover's response~$B$ in the presence of all information potentially exposed by the protocol\footnote{Although this particular conditional entropy is not commonly analyzed in standard device-independent settings, it has appeared in other contexts; see, for example,~\cite{fu2018localrandomness}.}.

        The optimization described above certifies a \emph{single-round} lower bound on the conditional min-entropy $H_{\min}(B \mid X{=}0, A{=}0, Y{=}0, E)$.
        Since $H_{\min}(\cdot) \leq H(\cdot)$, this also yields a lower bound on the conditional von Neumann entropy $H(B \mid X{=}0, A{=}0, Y{=}0, E)$.
        An alternative avenue is to target the von Neumann entropy directly, for example via the SDP framework of Brown et al.~\cite{brown2021randomness}.
        Adapting that objective to our computational-SoC constraints—particularly approximate nonsignaling with leakage and adaptive measurements—remains an interesting open question; we do not claim feasibility here.
        Either way, any per-round von Neumann bound is exactly what the Entropy Accumulation Theorem (EAT) requires to lift single-round guarantees to $n$-round smooth min-entropy.
        Moreover, prior work~\cite{merkulov2023entropy} established that EAT applies in the computational single-device setting we consider, so such bounds can be used directly as min-tradeoff functions for finite-size guarantees.

\section{Summary and open questions}

        We have introduced a framework that connects cryptographic assumptions to Bell inequalities via a canonical form protocol for single-prover interactive protocols and a Bell mapping that embeds them into a virtual Bell scenario.
        This allows us to define a computational analogue of the space of correlations, and to study classical and quantum behaviors under computational leakage using semidefinite programming.
        Our proposed computational-SoC hierarchy captures approximate nonsignaling constraints enforced through physically realizable measurements, and yields tight bounds on CHSH-like inequalities even in the presence of leakage.
        We demonstrate that this hierarchy subsumes known analytic bounds in the trapdoor-claw-free setting, and provides a pathway for translating device-independent tools to cryptographic protocols.

        We list several open questions.

        \paragraph{Minimal assumptions for quantum advantage.}
        The hidden-input condition is sufficient to upper bound the classical computational set~$\localcomp{\compLeakage}$.
        However, this condition alone need not imply a separation between the classical and quantum computational sets.
        A priori it could be that $\quantumcomp{\compLeakage}=\localcomp{\compLeakage}$.
        Section~\ref{sec:showcases} exhibits cryptographic assumptions under which a separation does occur.
        It is therefore natural to ask---What are the minimal assumptions that guarantee $\quantumcomp{\compLeakage}\not\subseteq\localcomp{\compLeakage}$?\footnote{A recent work~\cite{morimae2025owp} poses a related question in a different setting.}

        \paragraph{Information---disturbance trade-offs for virtual inputs.}
        Partial predictability of the virtual input charecterized by a computational leakage parameter~$\compLeakage$ does not automatically translate into the ability to realize points in~$\quantumcomp{\compLeakage}$.
        A QPT prover that tries to guess $\imap(\trans)$ may pay a coherence cost that degrades the quantum strategy it hopes to execute later.
        Can the trade-offs between virtual-input leakage and state disturbance be formalized?

        \paragraph{Randomness and entropy certification.}
        Can tools from device-independent randomness certification, such as those developed in~\cite{brown2021randomness}, be adapted to the computational setting?
        These techniques, originally designed for ideal Bell tests, will likely yield stronger entropy bounds or more noise-tolerant protocols when combined with our Bell mapping framework.

        \paragraph{Convergence of the hierarchy.}
        Can the convergence proof of the sequential NPA hierarchy in~\cite{kulpe2024boundquantumvaluecompiled} be extended to our setting?
        Unlike their model, which enforces negligible signaling on all monomials, our hierarchy imposes approximate nonsignaling only on sequentially implementable measurements and allows~$\O(\compLeakage)$ leakage.
        Determining whether convergence still holds under these relaxed constraints would clarify the long-term behavior of our hierarchy and its relationship to computational soundness.

        \paragraph{Protocol-specific analysis.}
        One can also approached the analysis from a more protocol-specific viewpoint.
        Instead of considering all canonical-form protocols that satisfy the hidden-input condition, it could be useful to fix a concrete protocol and its verifier and then study how variations in the hiding parameter~$\compLeakage$ affect its induced computational sets.
        Indeed, given a protocol with some degree of computational hiding, one can often amplify or reduce that hiding by simple transformations---for example, by repeating the protocol and extracting additional randomness to strengthen hiding, or by modifying the verifier to leak partial information about the virtual input to weaken it.
        Framing the analysis around a fixed verifier and controlled modifications of its leakage could lead to a finer understanding of how computational hiding interacts with provable separations between~$\localcomp{\compLeakage}$ and~$\quantumcomp{\compLeakage}$ for a given protocol.


        \paragraph{Acknowledgments.} 
            We thank Peter Brown for generous assistance with the \texttt{ncpol2sdpa} library and for helping resolve several technical issues, Thomas Hahn for helpful discussions and coding support, and Thomas Vidick for insightful discussions and suggestions. We thank Efrat Gerchkovitz for her careful reading of the manuscript and feedback.
            This research was generously supported by the Peter and Patricia Gruber Award, the Koshland Research Fund and the Air Force Office of Scientific Research under award number FA9550-22-1-0391.

\appendix

\section{Supplementary Material}\label{ap_sec:supp_material}

    \begin{lemma}[Continuity of the MDL set]\label{lemma:mdl_continuity}
        Let~$P\in\mdl\varepsilon\coloneqq\mdl{(l-\varepsilon,h+\varepsilon)}$.
        There exists~$Q\in\mdl{(l,h)}$ such that
        \begin{equation}
            \delta(P,Q) = O(\varepsilon) \; .
        \end{equation}
    \end{lemma}

    \begin{proof}
        We denote $\eta\coloneqq1/|\Ia|\cdot|\Ib|$.
        We for now assume $l<\eta<h$ and address the other cases later.
        Let $P\in\mdl\varepsilon$. There exists a decomposition of $P$ to hidden variables $\hp$ such that
        \begin{equation}
            P(\oa, \ob, \ia, \ib) = \int dg~P(\hp)P(\ia,\ib|\hp)P(\oa|\ia,\hp)P(\ob|\ib,\hp) \; .
        \end{equation}
        Now, we proceed to define the probability distribution~$Q\in\mdl{(l,h)}$ which will resemble~$P\in\mdl\varepsilon$ with a correction to the distributions~$P(\ia,\ib|\hp)$, via the uniform distribution, that will place it in~$\mdl{(l,h)}$.
        $Q$ will have the following decomposition to parameters~$\hp$
        \begin{equation}
            Q(\oa,\ob,\ia,\ib) = \int d\hp P(\hp) Q(\ia,\ib|\hp) Q(\oa,\ob|\ia,\ib,\hp)\;,
        \end{equation}
        where~$Q(\ia,\ib|\hp)\coloneqq~(1-q)P(\ia,\ib|\hp)+q\eta$ and $Q(\oa,\ob|\ia,\ib,\hp)\coloneqq P(\oa|\ia,\hp)P(\ob|\ib,\hp)$.
        To ensure~$Q\in\mdl{(l,h)}$, we want to choose~$q$ such that
        \begin{equation}
            (1-q)(h+\varepsilon)+q\eta\ \leq h
        \end{equation}
        and
        \begin{equation}
            (1-q)(l-\varepsilon)+q \eta\ \geq l\;.
        \end{equation}
        Denoting~$\mu\coloneqq\min\qty{h-\eta,\eta-l}$, both equations are satisfied by choosing
        \begin{equation}
            q = \frac{\varepsilon}{\mu+\varepsilon}=\frac{1}{\mu}\varepsilon+O(\varepsilon^2)\;.
        \end{equation}
        We now proceed to bound the variation distance between~$P$ and~$Q$.
        For all~$\oa,\ob,\ia,\ib$,
        \begin{align}
            & \qty|Q(\oa,\ob,\ia,\ib)-P(\oa,\ob,\ia,\ib)| \\
            = & \left| \sum_\hp P(\hp)\qty(Q(\ia,\ib|\hp)-P(\ia,\ib|\hp))P(\oa,\ob|\ia,\ib,\hp) \right| \\
            \leq & \sum_\hp P(\hp)\left| Q(\ia,\ib|\hp)-P(\ia,\ib|\hp) \right| \cdot1 \\
            =  & \sum_gP(\hp)\left| (1-q)P(\ia,\ib|\hp)+q\eta-P(\ia,\ib|\hp) \right| \\
            =  & \sum_gP(\hp)\left| -P(\ia,\ib|\hp)+\eta \right| q \\
            \leq & \sum_\hp P(\hp) q \\
            = & ~q \\
            = & O(\varepsilon)\;.
        \end{align}
        
        This covers the case where~$l<\eta<h$.
        For the remaining cases, we first note that~$\mdl{(\eta-r_0,\eta)}=\mdl{(\eta,\eta)}=\mdl{(\eta,\eta+r_1)}$ for any~$r_0,r_1\geq0$.
        W.L.O.G, we show~$\mdl{(\eta,\eta+r)}=\mdl{(\eta,\eta)}$.
        Assume towards a contradiction that given a distribution~$T\in\mdl{(\eta,\eta+r)}$, there exists~$(\ia_0,\ib_0,\hp)$ such that~$T(\ia_0,\ib_0|\hp)=\eta+\kappa$ for~$\kappa>0$.
        Then (recall that~$1/\eta$ is the number of elements~$(\ia,\ib)$),
        \begin{equation}
            1=\sum_{\ia,\ib}T(\ia,\ib|\hp)=\eta+\kappa+\sum_{(\ia,\ib)\neq(\ia_0,\ib_0)}T(\ia,\ib|\hp)\geq \eta+\kappa+\qty(\frac{1}{\eta}-1)\eta=1+\kappa\;.
        \end{equation}
        Therefore, we are to only be concerned with the continuity of the set~$\mdl{(\eta,\eta)}$.
        Hence, given a distribution~$\P\in\mdl{(\eta-\varepsilon,\eta+\varepsilon)}$, we choose the same distribution~$Q$ defined previously, with~$q=1$ and repeat similar steps to find~$\delta(P,Q)\leq \varepsilon$.
    \end{proof}

    \begin{lemma}[Continuity of the AMDL set]\label{lemma:amdl_continuity}
        Fix $\kappa\geq 0$ and let $\varepsilon>0$.
        If $P\in\amdl{\kappa+\varepsilon}$, then there exists $Q\in\amdl{\kappa}$ such that
        \begin{equation}
            \delta(P,Q)\;\le\;\frac{\varepsilon}{\kappa+\varepsilon}\;.
        \end{equation}
        In particular, for fixed $\kappa>0$, $\delta(P,Q)=O(\varepsilon)$ as $\varepsilon\to 0$.
    \end{lemma}

    \begin{proof}
        By membership $P\in\amdl{\kappa+\varepsilon}$, there exist a hidden-variable space $\Gamma$ with density $\hpDensity$, local response families
        $\{P_\gamma(\oa\mid\ia)\}_{\gamma\in\Gamma}$ and $\{P_\gamma(\ob\mid\ib)\}_{\gamma\in\Gamma}$, and an input law $P(\ia\mid\gamma)$ and marginal $P(\ib)$ such that
        \begin{equation}
            P(\oa,\ob,\ia,\ib)\;=\;\int d\gamma\,\hpDensity(\gamma)\,P(\ia\mid\gamma)\,P(\ib)\,P_\gamma(\oa\mid\ia)\,P_\gamma(\ob\mid\ib)\;,
        \end{equation}
        and
        \begin{equation}
            \E_\gamma\!\left[\max_{\ia}P(\ia\mid\gamma)-\frac{1}{|\Ia|}\right]\;\le\;\kappa+\varepsilon\;.
        \end{equation}
        Let $U(\ia)\coloneqq 1/|\Ia|$ be the uniform distribution on $\Ia$ and set
        \begin{equation}
            \alpha\;\coloneqq\;\frac{\varepsilon}{\kappa+\varepsilon}\in(0,1] \;.
        \end{equation}
        Define a “uniformly damped” input law
        \begin{equation}
            Q(\ia\mid\gamma)\;\coloneqq\;(1-\alpha)\,P(\ia\mid\gamma)\;+\;\alpha\,U(\ia)\;,
        \end{equation}
        and keep all other components unchanged:
        \begin{equation}
            Q(\ib)\coloneqq P(\ib)\;,\qquad
            Q_\gamma(\oa\mid\ia)\coloneqq P_\gamma(\oa\mid\ia)\;,\qquad
            Q_\gamma(\ob\mid\ib)\coloneqq P_\gamma(\ob\mid\ib)\;.
        \end{equation}
        Let $Q$ be the joint distribution generated by these choices:
        \begin{equation}
            Q(\oa,\ob,\ia,\ib)\;=\;\int d\gamma\,\hpDensity(\gamma)\,Q(\ia\mid\gamma)\,Q(\ib)\,Q_\gamma(\oa\mid\ia)\,Q_\gamma(\ob\mid\ib)\;.
        \end{equation}

        \emph{(i) $Q\in\amdl{\kappa}$.}
        For each $\gamma$,
        \begin{equation}
            \max_{\ia}Q(\ia\mid\gamma)-\frac{1}{|\Ia|}
            \;=\;\max_{\ia}\left((1-\alpha)P(\ia\mid\gamma)+\alpha U(\ia)\right)-U(\ia)
            \;\le\;(1-\alpha)\left(\max_{\ia}P(\ia\mid\gamma)-U(\ia)\right).
        \end{equation}
        Taking expectation and using $\E_\gamma[\max_{\ia}P(\ia\mid\gamma)-U(\ia)]\le\kappa+\varepsilon$,
        \begin{equation}
            \E_\gamma\!\left[\max_{\ia}Q(\ia\mid\gamma)-\frac{1}{|\Ia|}\right]
            \;\le\;(1-\alpha)\,(\kappa+\varepsilon)
            \;=\;\frac{\kappa}{\kappa+\varepsilon}\,(\kappa+\varepsilon)\;=\;\kappa\;,
        \end{equation}
        so $Q\in\amdl{\kappa}$.

        \emph{(ii) Total-variation bound.}
        Since $P$ and $Q$ differ only through replacing $P(\ia\mid\gamma)$ by $Q(\ia\mid\gamma)$, for any $(\oa,\ob,\ia,\ib)$,
        \begin{align}
            \left|Q(\oa,\ob,\ia,\ib)-P(\oa,\ob,\ia,\ib)\right|
            &=\left|\int d\gamma\,\hpDensity(\gamma)\,\bigl(Q(\ia\mid\gamma)-P(\ia\mid\gamma)\bigr)\,P(\ib)\,P_\gamma(\oa\mid\ia)\,P_\gamma(\ob\mid\ib)\right| \\
            &\le\int d\gamma\,\hpDensity(\gamma)\,\left|Q(\ia\mid\gamma)-P(\ia\mid\gamma)\right| \\
            &=\int d\gamma\,\hpDensity(\gamma)\,\alpha\,\left|U(\ia)-P(\ia\mid\gamma)\right| \\
            &\leq \alpha\;,
        \end{align}
        where we used $0\le P_\gamma(\cdot\mid\cdot),P(\ib)\le 1$ and $\left|U(\ia)-P(\ia\mid\gamma)\right|\le 1$ pointwise. Summing and dividing by $2$ gives
        \begin{equation}
            \delta(P,Q)\;\le\;\alpha\;=\;\frac{\varepsilon}{\kappa+\varepsilon}\;.
        \end{equation}
        This yields the claimed bound. In particular, for fixed $\kappa>0$ the right-hand side is $O(\varepsilon)$.
    \end{proof}

    \begin{lemma}[Continuity of Bell inequalities]\label{lemma:ineq_continuity}
        Let~$\I:\uc\rightarrow\mathbb{R}$ be an MDL inequality.
        Let~$P_0,P_1\in\uc$. Then,
        \begin{equation}\label{eq:ineq_continuity}
            \I(P_0) = \I(P_1) + \O(\delta(P_0,P_1)) \; .
        \end{equation}
    \end{lemma}

    \begin{proof}
        \begin{align}
            \qty|\I(P_0)-\I(P_1)| & = \sum_{\oa,\ob,\ia,\ib} \qty|v_{\oa\ob\ia\ib}\cdot\qty(P_0-P_1)(\oa,\ob,\ia,\ib)| \\
                & \leq \sum_{\oa,\ob,\ia,\ib} \qty|v_{\oa\ob\ia\ib}|\cdot\delta(P_0,P_1) \\
                & = \O(\delta(P_0,P_1)) \; .
        \end{align}
    \end{proof}

    \begin{proof}[Proof of Theorem~\ref{thm:cl_to_local}]
        \label{proof:cl_to_local}
            Let~$\pcomp\in\localcomp{\compLeakage}$.
            By Lemma~\ref{lemma:classical_amdl_reduction}, $\localcomp{\compLeakage}$ is a subset of the closure of $\amdl{\compLeakage}$.
            Apply Lemma~\ref{lemma:amdl_mdl_inequality_bound} with~$\local_\flexibility = \mdl{(\mdlparams)}$
            where
            \begin{equation}
                \mdlparamsExplicit \; .
            \end{equation}
            Thus, we obtain an MDL inequality $\I_\flexibility$ valid for $\local_\flexibility$ such that, for any $P\in\amdl{\compLeakage}$,
            \begin{equation}
                \I_\flexibility(P)\;\leq\; \frac{\compLeakage}{\compLeakage+\flexibility}\cdot
                \max_{S\in\uc}\I_\flexibility(S)
                \;=\; \O\!\left(\frac{\compLeakage}{\compLeakage+\flexibility}\right).
            \end{equation}
            Let $c_\flexibility \coloneqq \max_{S\in\uc}\I_\flexibility(S)$ and define the shifted functional
            \begin{equation}
                \I(P)\;\coloneqq\; \I_\flexibility(P)\;-\; c_\flexibility\cdot
                \frac{\compLeakage}{\compLeakage+\flexibility}\;.
            \end{equation}
            Since $\I_\flexibility$ is valid for $\local_\flexibility$, it follows that $\I(L)\le 0$ for all $L\in\local_\flexibility$.
            Moreover, by the bound above and continuity under closure, $\I(\pcomp)\le 0$.
            Hence $\I$ is the desired computational Bell inequality for~$\localcomp{\compLeakage}$ .
        \end{proof}

    \begin{proof}[Proof of Lemma~\ref{lemma:amdl_mdl_decomposition_one_sided}]
    \label{proof:amdl_mdl_decomposition_one_sided}
            Define $M(\gamma)\coloneqq \max_{\ia}P(\ia\mid\gamma)-\frac{1}{|\mathcal{X}|}$.
            By Markov's inequality,
            \begin{equation}
                \Pr_\gamma\left[M(\gamma)>\kappa+\vartheta\right]\le \frac{\kappa}{\kappa+\vartheta} \; .
            \end{equation}
            Let $\Gamma_\vartheta\coloneqq\{\gamma:\,M(\gamma)\le \kappa+\vartheta\}$ and decompose
            \begin{equation}
                P \;=\; (1-\alpha)\,L \;+\; \alpha\,S \; , \qquad
                \alpha \;=\; \Pr(\gamma\notin\Gamma_\vartheta)\;\le\; \frac{\kappa}{\kappa+\vartheta} \; ,
            \end{equation}
            where $L$ (resp. $S$) is the conditional distribution given $\gamma\in\Gamma_\vartheta$ (resp. $\gamma\notin\Gamma_\vartheta$).

            For any $\gamma\in\Gamma_\vartheta$ we have
            \begin{equation}
                \max_{\ia} P(\ia\mid\gamma) \;\le\; \frac{1}{|\mathcal{X}|}+\kappa+\vartheta \; .
            \end{equation}
            Hence, using $P(\ib)=1/|\mathcal{Y}|$ and $P(\ia,\ib\mid\gamma)=P(\ia\mid\gamma)P(\ib)$,
            \begin{equation}
                P(\ia,\ib\mid\gamma) \;\leq\; \frac{1}{|\mathcal{Y}|}\!\left(\frac{1}{|\mathcal{X}|}+\kappa+\vartheta\right) \; .
            \end{equation}
            For the lower bound, for any distribution on $|\mathcal{X}|$ points we have
            \begin{equation}
                \min_{\ia} P(\ia\mid\gamma) \;\ge\; 1 - (|\mathcal{X}|-1)\,\max_{\ia} P(\ia\mid\gamma) \; .
            \end{equation}
            Therefore, for $\gamma\in\Gamma_\vartheta$,
            \begin{equation}
                \min_{\ia} P(\ia\mid\gamma) \;\ge\; 1 - (|\mathcal{X}|-1)\left(\frac{1}{|\mathcal{X}|}+\kappa+\vartheta\right) \; ,
            \end{equation}
            and thus
            \begin{equation}
                P(\ia,\ib\mid\gamma) \;\ge\; \frac{1}{|\mathcal{Y}|}\left(1 - (|\mathcal{X}|-1)\left(\frac{1}{|\mathcal{X}|}+\kappa+\vartheta\right)\right) \; .
            \end{equation}
            This shows that the conditional component $L$ lies in the MDL set~$\mdl{(l_\flexibility,h_\flexibility)}$.
            The claimed decomposition follows.
        \end{proof}

        \begin{proof}[Proof for Lemma~\ref{lemma:amdl_mdl_inequality_bound}]
        \label{proof:amdl_mdl_inequality_bound}
            By Lemma~\ref{lemma:amdl_mdl_decomposition_one_sided}, any $P\in\amdl{\kappa}$ can be decomposed as
            \begin{equation}
                P \;=\; \left(1-\alpha\right)L + \alpha S \; ,
            \end{equation}
            with $\alpha \leq \kappa/(\kappa+\flexibility)$, $L\in\local_\flexibility$, and $S\in\uc$.
            
            Since $\bellInequality_\flexibility$ is affine,
            \begin{equation}
                \bellInequality_\flexibility(P)
                \;=\; (1-\alpha)\,\bellInequality_\flexibility(L) + \alpha\,\bellInequality_\flexibility(S) \; .
            \end{equation}
            By validity of the inequality, $\bellInequality_\flexibility(L)\leq 0$.
            Hence
            \begin{equation}
                \bellInequality_\flexibility(P)
                \;\le\; \alpha\cdot \max_{S\in\uc}\;\bellInequality_\flexibility(S)
                \;\le\; \frac{\kappa}{\kappa+\flexibility}\cdot \max_{S\in\uc}\;\bellInequality_\flexibility(S) \; ,
            \end{equation}
            which is $O(\kappa/(\kappa+\flexibility))$ as claimed.
        \end{proof}

\section{Examples of basic concrete protocol}\label{ap_sec:protocol}

    \begin{proof}[Proof for Lemma~\ref{lemma:showcase_ineq}]
    \label{proof:showcase_ineq}
        By Lemma~\ref{lemma:classical_amdl_reduction}, any Bell-mapped distribution~$P_\security$ induced by a classical prover belongs to the closure of~$\amdl{\compLeakage}$. 
        By Lemma~\ref{lemma:amdl_mdl_inequality_bound}, for every $P \in \amdl{\compLeakage}$ we have
        \begin{equation}
            \I_\flexibility(P) \;\le\;
            \frac{\compLeakage}{\compLeakage+\flexibility}\cdot
            \max_{S\in\uc}\;\I_\flexibility(S) \;.
        \end{equation}
        The functional $\I_\flexibility$ is maximized by setting all mass on $(\oa,\ob,\ia,\ib)=(0,0,0,0)$, which yields
        \begin{align}
            \I_\flexibility (S) & \leq
            \tfrac{1}{2}\left(\tfrac{1}{2} - \compLeakage - \flexibility \right) S(A=0,B=0,X=0,Y=0) \\
            & \leq \tfrac{1}{2}\left(\tfrac{1}{2} - \compLeakage - \flexibility \right) S(Y=0) \\
            & \leq \tfrac{1}{2}\left(\tfrac{1}{2} - \compLeakage - \flexibility \right) \tfrac{1}{2} \; .
        \end{align}
        Therefore, for every $P\in\amdl{\compLeakage}$,
        \begin{equation}
            \I_\flexibility(P) \leq 
            \frac{1}{4}\frac{\compLeakage}{\compLeakage+\flexibility}
            \qty(\tfrac{1}{2}-\compLeakage-\flexibility) \;.
        \end{equation}
        Combining with the negligible error from the closure argument, we obtain for every~$\security$,
        \begin{equation}
            \I_\flexibility(P_\security) \leq
            \frac{1}{4}\frac{\compLeakage}{\compLeakage+\flexibility}
            \qty(\tfrac{1}{2}-\compLeakage-\flexibility) + \negl(\security) \;.
        \end{equation}
        Subtracting this worst-case bound from $\I_\flexibility$ as in Equation~\eqref{eq:comp_ineq_chsh_scenario}, we conclude
        \begin{equation}
            \I(P_\security) \leq \negl(\security) \;,
        \end{equation}
        which proves the claim.
    \end{proof}

    \rnote{In each subsection you should write a few words about the type of protocol, refer to the previous papers at least.}
    
    \subsection{Protocol based on trapdoor claw-free function}\label{ap_sec:kcvy_protocol}

        We present here the TCF-based protocol from~\cite{kahanamoku2022classically}. 
        This serves as our first showcase, and we adapt the protocol to the canonical form used in our framework. 
        The underlying primitive is a \emph{trapdoor claw-free function family}, formally defined in Definition~\ref{def:trapdoor_claw_free}. 
        The corresponding honest implementation, expressed in our canonical structure, is illustrated in Figure~\ref{fig:kcvy_protocol}.

        \begin{definition}[Trapdoor claw-free function family]\label{def:trapdoor_claw_free}
            A family of functions~$\mathcal{F} = \{ f_k : \mathcal{\Preimage} \to \mathcal{\Image} \}_{k \in \mathcal{K}}$ is called a \emph{trapdoor claw-free function family} if the following conditions hold:
            \begin{enumerate}
                \item \textbf{Key generation.}
                There exists a randomized polynomial-time algorithm~$\mathtt{Gen}(1^\security)$ that outputs a key-trapdoor pair~$(k, t)$, where~$k \in \mathcal{K}$ is a public key and~$t$ is a trapdoor.

                \item \textbf{Efficient evaluation.}
                There exists a deterministic polynomial-time algorithm that, given~$k \in \mathcal{K}$ and~$\preimage \in \mathcal{\Preimage}$, computes~$f_k(\preimage)$.

                \item \textbf{Efficient inversion with trapdoor.}
                There exists a deterministic polynomial-time algorithm that, given a~$t$ and~$\image \in \mathcal{\Image}$ such that~$\image = f_k(\preimage)$ for some~$\preimage \in \mathcal{\Preimage}$, recovers the full claw.
                I.e., there are exactly two preimages~$\preimage_0$ and~$\preimage_1$ such that~$f_k(\preimage_0) = f_k(\preimage_1) = \image$ and~$\preimage_0 \ne \preimage_1$.

                \item \textbf{Claw-freeness.}
                For every probabilistic polynomial-time (PPT) adversary~$\A$, the probability that~$\A(k)$ outputs a \emph{claw} is negligible in~$\security$:
                \begin{equation}
                    \Pr_{(k,t) \leftarrow \mathtt{Gen}(1^\security)} \left[
                        \begin{array}{c}
                            (\preimage_0, \preimage_1) \leftarrow \A(k) \\
                            \text{s.t. } \preimage_0 \ne \preimage_1 \text{ and } f_k(\preimage_0) = f_k(\preimage_1)
                        \end{array}
                    \right] \leq \negl(\security) \;.
                \end{equation}
                That is, it is computationally hard to find two distinct preimages~$\preimage_0, \preimage_1$ that collide under~$f_k$, even though the claw exists.
            \end{enumerate}
        \end{definition}

        \begin{figure}[ht]
            \centering
            \includegraphics[width=0.85\textwidth]{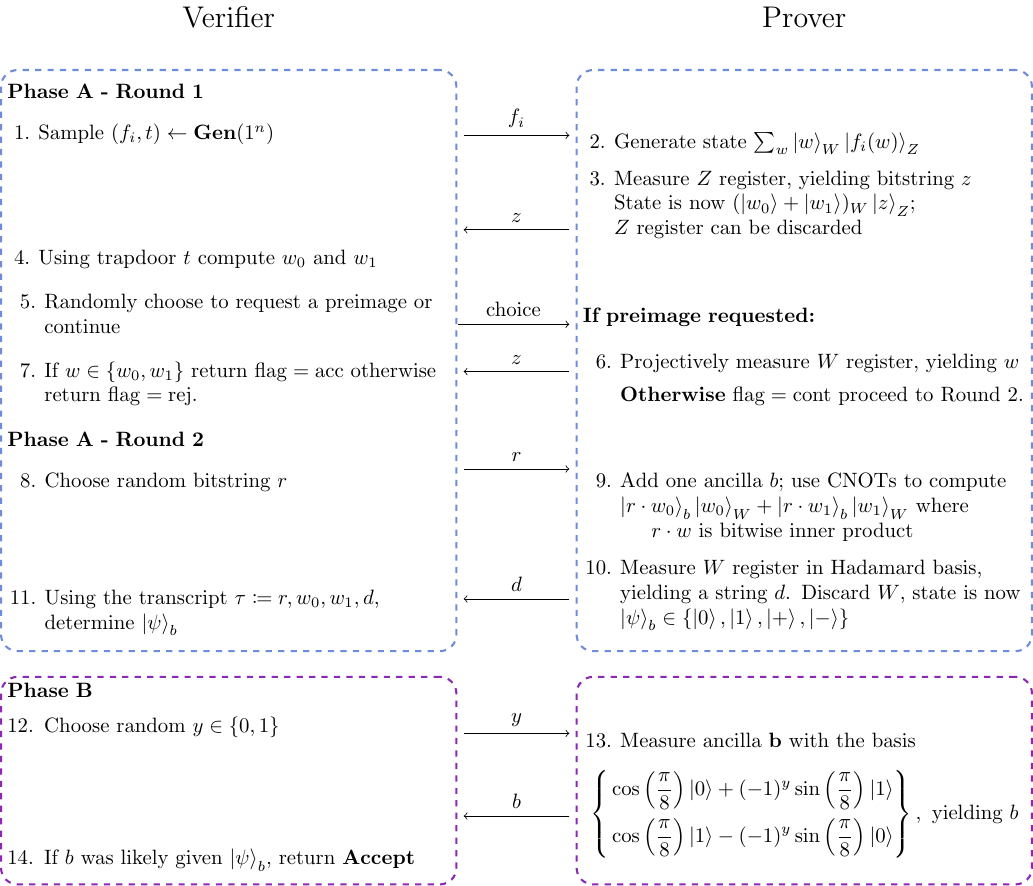}
            \caption{\footnotesize Honest implementation of the TCF based protocol canonical form.
            This figure is adapted from~\cite[Figure 1]{kahanamoku2022classically} to match the canonical protocol structure used in our framework.}
            \label{fig:kcvy_protocol}
        \end{figure}

    \subsection{Protocol based on a compiled game}\label{ap_sec:compiled}
        We present here the compiled nonlocal game protocol from~\cite{kalai2022compiled,compiledtrapdoor2024}. 
        This serves as our second showcase, and we adapt the protocol to the canonical form used in our framework.
        
        \begin{definition}[Quantum Homomorphic Encryption (QHE)]\label{def:QHE}
            A quantum homomorphic encryption scheme $\mathrm{QHE} = (\Gen, \Enc, \Eval, \Dec)$ for a class of quantum circuits~$\mathcal{C}$ is a tuple of algorithms with the following syntax:
            \begin{itemize}
                \item $\Gen$ is a PPT algorithm that takes as input the security parameter~$1^\security$ and outputs a (classical) secret key~$\mathsf{sk}$ of $\mathrm{poly}(\security)$ bits.
                \item $\Enc$ is a PPT algorithm that takes as input a secret key~$\mathsf{sk}$ and a classical input~$\ia$, and outputs a ciphertext~$\mathsf{ct}$.
                \item $\Eval$ is a QPT algorithm that takes as input a tuple $(C, \ket{\psi}, \mathsf{ct}_{\mathrm{in}})$, where
                \begin{itemize}
                    \item $C : \mathcal{H}_A \otimes (\mathbb{C}^2)^{\otimes n} \to (\mathbb{C}^2)^{\otimes m}$ is a quantum circuit,
                    \item $\ket{\psi} \in \mathcal{H}_A$ is a quantum state, and
                    \item $\mathsf{ct}_{\mathrm{in}}$ is a ciphertext corresponding to an $n$-bit plaintext.
                \end{itemize}
                $\Eval$ computes $\mathsf{ct}_{\mathrm{out}} \leftarrow \Eval_C(\ket{\psi}, \mathsf{ct}_{\mathrm{in}})$ and outputs a ciphertext~$\mathsf{ct}_{\mathrm{out}}$. If $C$ has classical output, then $\Eval_C$ must also produce classical output.
                \item $\Dec$ is a PT algorithm that takes as input the secret key~$\mathsf{sk}$ and a ciphertext~$\mathsf{ct}$, outputting a quantum state~$\ket{\varphi}$. If $\mathsf{ct}$ encodes a classical message, then $\Dec$ outputs a classical string~$\ib$.
            \end{itemize}
            The following properties are required:
            \begin{enumerate}
                \item \textbf{Correctness with Auxiliary Input.} For every~$\security \in \mathbb{N}$, every circuit $C : \mathcal{H}_A \otimes (\mathbb{C}^2)^{\otimes n} \to \{0,1\}^*$, every quantum state $\ket{\psi}_{AB} \in \mathcal{H}_A \otimes \mathcal{H}_B$, message $\ia \in \{0,1\}^n$, key $\mathsf{sk} \leftarrow \Gen(1^\security)$, and ciphertext $\mathsf{ct} \leftarrow \Enc(\mathsf{sk}, \ia)$, the outputs of the following two experiments are negligibly close in trace distance:
                \begin{description}
                    \item[Game 1.] Start with $(\ia, \ket{\psi}_{AB})$. Evaluate $C$ on $\ia$ and register~$A$, producing a classical output $\ib$, and output $(\ib, \mathrm{reg}_B)$.
                    \item[Game 2.] Start with $\mathsf{ct} \leftarrow \Enc(\mathsf{sk}, \ia)$ and $\ket{\psi}_{AB}$. Compute $\mathsf{ct}' \leftarrow \Eval_C(\ket{0}^{\otimes \mathrm{poly}(\security,n)}, \mathsf{ct})$ on register~$A$. Compute $\ib' = \Dec(\mathsf{sk}, \mathsf{ct}')$. Output $(\ib', \mathrm{reg}_B)$.
                \end{description}
        
                \item \textbf{T-Classical Security.} For any two messages $\ia_0, \ia_1$ and any classical circuit ensemble~$\mathcal{A}$ of size $\mathrm{poly}(T(\security))$,
                \begin{equation}
                    \left|\Pr\left[\mathcal{A}(\mathsf{ct}_0)=1\right] - \Pr\left[\mathcal{A}(\mathsf{ct}_1)=1\right]\right|
                    \leq \negl\left(T(\security)\right),
                \end{equation}
                where $\mathsf{sk}\leftarrow \Gen(1^\security)$, and $\mathsf{ct}_i\leftarrow\Enc(\mathsf{sk}, \ia_i)$ for $i=0,1$.
            \end{enumerate}
        \end{definition}
    
        \begin{figure}[ht]\label{fig:compiled_protocol}
            \centering
            \includegraphics[width=0.85\textwidth]{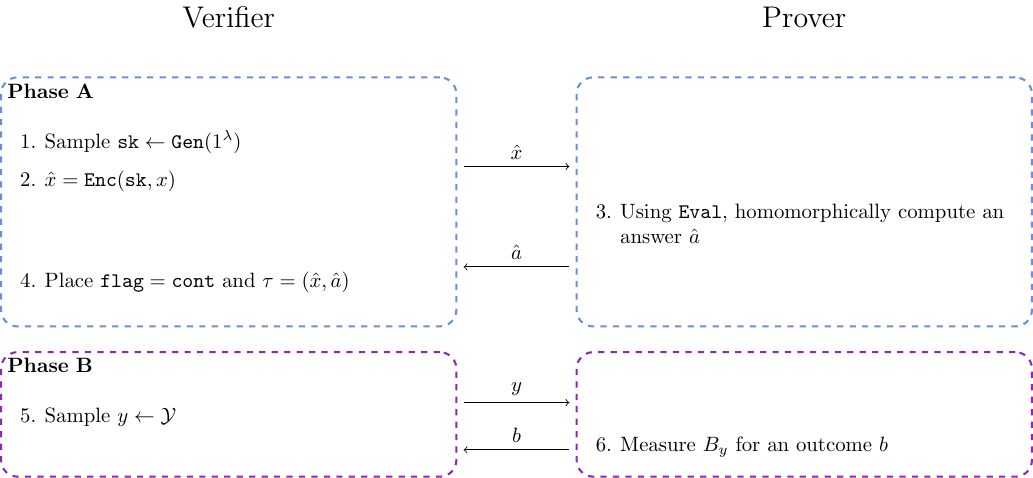}
            \caption
            {
                \footnotesize
                Honest implementation of the compiled nonlocal game from~\cite{kalai2022compiled}.
                The translation to canonical form highlights the natural decomposition of the interaction into classical preprocessing (Phase A) and measurement-based response (Phase B).
                This protocol requires a quantum homomorphic encryption scheme $\QHE = (\QHEGen, \QHEEnc, \QHEEval, \QHEDec)$.
            }
            \label{fig:kalai_protocol}
        \end{figure}

    \section{Polytopality of AMDL}

        \newcommand{\polyhedron}{\Theta}
        \newcommand{\dethp}{s}
        \newcommand{\detHp}{\mathcal{\MakeUppercase{\dethp}}}
        \begin{theorem}[Polytope structure of $\amdl{\kappa}$]\label{lemma:amdl_polytope}
            Fix a finite Bell scenario $\bellscenarioShort=(\Ia,\Ib,\Oa,\Ob)$ and $\kappa\geq 0$. 
            Then the set $\amdl{\kappa}$ of joint distributions $P(\oa,\ob,\ia,\ib)$ is a polytope.
        \end{theorem}
                
        \begin{proof}
            By standard arguments, we may take the local responses to be deterministic without loss of generality. 
            Let $\mathcal{F}_\playerA:=\Oa^{\Ia}$ and $\mathcal{F}_\playerB:=\Ob^{\Ib}$. 
            Let $\detHp:=\mathcal{F}_\playerA\times\mathcal{F}_\playerB$, and write $\dethp=(f_\playerA,f_\playerB)\in\detHp$.
            
            Introduce nonnegative variables $w_{\dethp,\ia}$ and $r_\dethp$ for each $\dethp\in\detHp$ and $\ia\in\Ia$. 
            Impose the linear constraints
            \begin{align}
                &\sum_{\dethp\in\detHp}\sum_{\ia\in\Ia} w_{\dethp,\ia} \;=\; 1 \;, \label{eq:norm}\\
                &\forall \dethp\in\detHp,\ \forall \ia\in\Ia:\quad r_\dethp \;\geq\; w_{\dethp,\ia} \;, \label{eq:rsgeq}\\
                &\sum_{\dethp\in\detHp} r_\dethp \;\leq\; \frac{1}{|\Ia|}+\kappa \;, \label{eq:kappabound}\\
                &\forall \oa,\ob,\ia,\ib:\quad
                P(\oa,\ob,\ia,\ib)
                \;=\; \frac{1}{|\Ib|}\sum_{\dethp=(f_\playerA,f_\playerB)\in\detHp} w_{\dethp,\ia}\ \indicator{\oa=f_\playerA(\ia)}\ \indicator{\ob=f_\playerB(\ib)}\;. \label{eq:Pfromw}
            \end{align}
            
            Let~$\polyhedron$ be the set of triples $(P,w,r)$ satisfying \eqref{eq:norm}--\eqref{eq:Pfromw}. 
            This is a polyhedron since all constraints are linear. 
            Its projection onto distributions $P$ is therefore a polyhedron.
            
            We show equality between this projection and $\amdl{\kappa}$.
            
            ($\subseteq$) 
            Given $(P,w,r)\in\polyhedron$, define a hidden variable that first samples $\dethp\in\detHp$ with probability $\lambda_\dethp:=\sum_{\ia}w_{\dethp,\ia}$. 
            Conditioned on $\dethp$, sample $\ia$ with probability $P(\ia\mid \dethp)=w_{\dethp,\ia}/\lambda_\dethp$, and sample $\ib$ uniformly. 
            Output $\oa=f_\playerA(\ia)$ and $\ob=f_\playerB(\ib)$. 
            Then \eqref{eq:Pfromw} gives exactly $P(\oa,\ob,\ia,\ib)$. 
            Moreover,
            \begin{equation}
                \E\left[\max_{\ia}P(\ia\mid \dethp)\right]
                \;=\; \sum_{\dethp}\lambda_\dethp \max_{\ia}\frac{w_{\dethp,\ia}}{\lambda_\dethp}
                \;=\; \sum_{\dethp}\max_{\ia} w_{\dethp,\ia}
                \;\leq\; \sum_{\dethp} r_\dethp
                \;\leq\; \tfrac{1}{|\Ia|}+\kappa \;,
            \end{equation}
            using \eqref{eq:rsgeq} and \eqref{eq:kappabound}. 
            Hence $P\in\amdl{\kappa}$.
            
            ($\supseteq$) 
            Conversely, take any $P\in\amdl{\kappa}$ witnessed by a distribution $\hpDensity$:
            \begin{equation}
                P \;=\; \frac{1}{|\Ib|}\int d\hp\;\hpDensity(\hp)\;P(\ia\mid\hp)\;P(\oa\mid\ia,\hp)P(\ob\mid\ib,\hp)\;.
            \end{equation}
            \cite[Theorem 2.1]{scarani2019bell} allows us to write the local distributions as a convex sum of deterministic ones.
            \begin{align}
                P & \;=\; \frac{1}{|\Ib|}\int d\hp\;\hpDensity(\hp)\;P(\ia\mid\hp)\;
                \sum_{\dethp}P(\dethp\mid\hp)\,\indicator{f_\playerA^{\dethp}(x)=a}\,\indicator{f_\playerB^{\dethp}(y)=b}\\
                 & \;=\; \frac{1}{|\Ib|} \sum_{\dethp}
                 \left(\int d\hp\;P(\dethp,\hp)\;P(\ia\mid\hp)\right)
                \indicator{f_\playerA^{\dethp}(x)=a}\,\indicator{f_\playerB^{\dethp}(y)=b}\;.
            \end{align}
            
            We denote
            \begin{equation}
                w_{\dethp,\ia}\coloneqq\int d\hp\;P(\dethp,\hp)\;P(\ia\mid\hp),\qquad
                r_\dethp\coloneqq\int d\hp\;P(\dethp,\hp)\;\max_\ia P(\ia\mid\hp)\;,
            \end{equation}
            which defines a valid tuple in the set of triples~$\polyhedron$.
            Therefore,~$P$ belongs to the projection of~$\polyhedron$ on the set of distributions.
        \end{proof}

\bibliographystyle{unsrt}
\bibliography{bib}

\end{document}